\documentclass[letterpaper]{article} 
\usepackage{aaai23}  
\usepackage{times}  
\usepackage{helvet}  
\usepackage{courier}  
\usepackage[hyphens]{url}  
\usepackage{graphicx} 
\usepackage{natbib}  
\usepackage{caption} 
\usepackage[utf8]{inputenc} 
\usepackage[T1]{fontenc}    
\usepackage{hyperref}    
\usepackage{url}  
\usepackage{booktabs} 
\usepackage{amsfonts}
\usepackage{nicefrac}
\usepackage{microtype}
\usepackage{xcolor} 

\usepackage{graphicx}
\usepackage{dcolumn}
\usepackage{bm}
\usepackage{amsmath,amsthm}
\usepackage{braket}
\usepackage{mathtools}
\usepackage[caption=false]{subfig}
\usepackage{amssymb}
\usepackage{tikz}
\usetikzlibrary{quantikz}

\usepackage{dsfont}
\DeclareMathOperator*{\argmax}{arg\,max}

\usepackage{algorithm}
\usepackage{algorithmic}
\newtheorem{theorem}{Theorem}

\newtheorem{corollary}{Corollary}
\newtheorem{lemma}{Lemma}
\newtheorem{definition}{Definition}

\newcommand\numberthis{\addtocounter{equation}{1}\tag{\theequation}}

\newtheorem{innercustomthm}{Theorem}
\newenvironment{customthm}[1]
  {\renewcommand\theinnercustomthm{#1}\innercustomthm}
  {\endinnercustomthm}

\newtheorem{innercustomco}{Corollary}
\newenvironment{customco}[1]
  {\renewcommand\theinnercustomco{#1}\innercustomco}
  {\endinnercustomco}

\usepackage{newfloat}
\usepackage{listings}
\DeclareCaptionStyle{ruled}{labelfont=normalfont,labelsep=colon,strut=off} 
\floatstyle{ruled}
\newfloat{listing}{tb}{lst}{}
\floatname{listing}{Listing}
\nocopyright

\title{On the Trainability and Classical Simulability of Learning Matrix Product States Variationally}
\author{
    Afrad Basheer\footnote{Email: Afrad.M.Basheer@student.uts.edu.au}\textsuperscript{\rm 1},
    Yuan Feng\textsuperscript{\rm 2},
    Christopher Ferrie\textsuperscript{\rm 1},
    Sanjiang Li\textsuperscript{\rm 1} and 
    Hakop Pashayan\textsuperscript{\rm 3}
}
\affiliations{
    \textsuperscript{\rm 1}Centre for Quantum Software and Information, University of Technology Sydney, NSW 2007, Australia\\
    \textsuperscript{\rm 2} Department of Computer Science and Technology, Tsinghua University, Beijing 100084, China\\
    \textsuperscript{\rm 3} Dahlem Center for Complex Quantum Systems, Freie Universit\"{a}t Berlin 14195, Germany\\

}

\usepackage{bibentry}

\begin{document}

\maketitle

\begin{abstract}

We prove that using global observables to train the matrix product state ansatz results in the vanishing of all partial derivatives, also known as barren plateaus, while using local observables avoids this. This ansatz is widely used in quantum machine learning for learning weakly entangled state approximations. Additionally, we empirically demonstrate that in many cases, the objective function is an inner product of almost sparse operators, highlighting the potential for classically simulating such a learning problem with few quantum resources. All our results are experimentally validated across various scenarios. 

\end{abstract}

\section{Introduction}

    Quantum computing has made major strides in the past few years in terms of qubit capacity \cite{Chow2021,Gambetta2022} and error correction \cite{Dasilva2024}, prompting significant efforts in demonstrating concrete advantages over classical computers. Although this has been achieved in some cases~\cite{Arute2019,Zhong2020,Madsen2022}, superiority in solving practically useful problems remains to be seen. With the emergence of Noisy Intermediate Scale Quantum (NISQ) devices~\cite{Preskill2018}, recent focus has shifted to quantum optimization algorithms, especially Variational Quantum Algorithms (VQAs)~\cite{Cerezo2021_VQA}. In VQAs, we use quantum computers to estimate objective functions involving quantum \textit{states} and parameterized quantum circuits (also called \textit{ansatzes}) and update the parameters of these functions classically. Popular examples include variational quantum eigensolver~\cite{Peruzzo2014}, quantum support vector machines~\cite{Havlicek2019}, quantum approximate optimization algorithm~\cite{Farhi2014}, etc.
    
    An important example that features extensively in quantum information is learning cost-effective approximations of quantum states~\cite{Cerezo2021,Matos2021,Gard2020}. These techniques can be used to learn ansatzes that can approximately prepare states using fewer gates than initially required. Among the many choices of ansatzes used, interest has surged in the Matrix Product State (MPS) ansatz (cf. Figure~\ref{fig:ansatz} (a))~\cite{Rudolph2022,Ran2019,Dov2022,Lin2021,Rieser2023}. The MPS ansatz leverages the MPS data structure, which stores quantum states with space complexity that scales polynomially with the \textit{bond dimensions} (parameters that measure entanglement in linear qubit arrays). Consequently, the MPS ansatz is particularly effective for learning weakly entangled approximations, potentially resulting in fewer gate counts and simpler gate connectivity requirements compared to other approaches. 

    Learning state approximations variationally can be performed using either global or local observables~\cite{Cerezo2021}. Experimental results in~\cite{Dov2022} showed that using the MPS ansatz along with global observables for state approximations can result in barren plateaus, where all partial derivatives become exponentially small in the number of qubits. This makes estimating these derivatives using quantum devices require exponentially many executions. Moreover, the parameter updates also become exponentially small. In contrast, it was empirically shown that using local observables can help mitigate this issue.

    Theoretically proving this phenomenon is a challenging task. Although this has been achieved for similar problems such as optimization over the Hardware Efficient Ansatz (HEA)~\cite{Cerezo2021} (cf. Figure~\ref{fig:ansatz} (b)) and tensor network-based optimization in quantum information~\cite{Liu2022}, these results cannot be used to explain exponentially vanishing objective functions and gradients for the MPS ansatz.
    
    This work provides rigorous trainability proofs for the MPS ansatz. We prove that under uniformly random initialization of the circuit parameters, when using global observables, the variance of the objective function decreases exponentially while the usage of local observables ensures that the same variance is lower bounded by a quantity whose dependence on the number of qubits is linear and scales exponentially only in the width of the subcircuit involved. We also relate this with the variance of the partial derivatives and show that similar results hold for them as well.

    Trainability is closely interrelated with classical simulability. In~\cite{Cerezo2024}, it was conjectured, with evidence, that provably avoiding barren plateaus in this manner could imply classical simulability with few quantum resources. That is, for all provably trainable VQA objective functions, one can simulate the whole optimization classically using the outputs of a few quantum measurements implemented beforehand on the input state. By proving the trainability of the MPS ansatz and local observable combination, our work prepares the groundwork for studying its classical simulability.

    On this side, we demonstrate that these trainable VQA objective functions exhibit \textit{effective subspaces}. These subspaces are loosely defined as the subspaces where the observables, when conjugated with the ansatzes, tend to be mostly concentrated, for almost all input parameters~\cite{Cerezo2024}. If the objective function exhibits this property, then most function evaluations, which are nothing but inner products of the state with these conjugated observables (cf. Eq~\eqref{eq:main_optimization}), could potentially be classically estimated using the input state's coefficients in this subspace estimated beforehand using a quantum device. We first characterize the property of exhibiting effective subspaces by introducing an efficiently estimable norm for observables, the $C$-$\mathbb{K}$ norm, which we use to experimentally show that the MPS ansatz and local observable combination exhibits an effective subspace within the Pauli basis. 
    
    Our main contributions can be summarized as follows: 

    \begin{enumerate}
        \item For the problem of learning weakly entangled state approximations variationally, we rigorously prove that the usage of global observables will induce barren plateaus, while the usage of local observables will avoid them.
        \item We empirically show that the MPS ansatz, when used in combination with local observables, exhibits an effective subspace within the Pauli basis, which is conjectured to be a consequence of avoiding cost concentration and a sufficient condition for the ansatz to be classically simulable using few quantum resources as per~\cite{Cerezo2024}.
    \end{enumerate}
    Finally, we experimentally validate our results across various scenarios, including the impact of observable choices on MPS ansatz optimization and the detection of effective subspaces in MPS ansatz as well as other ansatzes such as HEA, and Quantum Convolutional Neural Network (QCNN) (cf. Figure~\ref{fig:ansatz} (c)).  
\section{Related Works} \label{sec:related_works}

    In \cite{Liu2022,Garcia2023,Barthel2024}, the theoretical study of barren plateaus in tensor-network-based machine learning with MPS inputs reveals that using global observables in the objective function introduces barren plateaus, whereas local observables avoid them. However, as mentioned in~\cite{Liu2022}, their model and assumptions differ from a variational circuit model. They model the input using the unitary decomposition of MPS, where each component tensor is reshaped into a $2D \times 2D$ unitary matrix, with $D$ as the bond dimension. The randomness is modeled by assuming these unitaries form \textit{unitary $2$-designs}~\cite{Dankert2009}, which are ensembles of unitaries such that integrating quadratic polynomials over them is equivalent to integrating the same over the Haar measure. In contrast, we assume that the subcircuits are sampled from unitary $2$-designs, which is more natural for circuit-based problems as a circuit depth polynomial in the width of the subcircuits suffices for them to behave like a unitary $2$-design under uniformly random parameter initialization \cite{Harrow2009}.

    In~\cite{Dov2022}, it was experimentally observed that while using the MPS ansatz, the usage of global observables leads to exponentially decreasing gradient magnitudes, whereas local observables avoid this issue. In our work, we study this phenomenon theoretically as well as similar trends in cost concentration. The existence of exponentially decreasing partial derivatives in MPS ansatz-based VQAs is proved using ZX-calculus in~\cite{Zhao2021}. However, the method can only be used to prove this for individual examples of the MPS ansatz, with pre-defined structures for the subcircuits. In contrast, our proofs consider the most generalized form of the ansatz, with the only assumption being that the subcircuits form unitary $2$ designs. Also in~\cite{Zhao2021}, there are no discussions regarding the impact that observables and subcircuit widths can have on trainability, which we theoretically demonstrate in the case of cost concentration as well as barren plateaus. 
    
\section{Background} \label{sec:background}
    We denote column vectors as $\ket{\psi}$ ('ket' notation) and their conjugate transposes as $\bra{\psi}$ ('bra' notation). The vector $\ket{i} \in \mathbb{C}^d$ represents the $i^{\text{th}}$ computational basis vector. We also define $\ket{\mathbf{0}} \coloneqq \ket{0} ^ {\otimes n}$, where $n$ is the number of qubits involved. We define $\mathds1_{t}$ to be the identity matrix acting on $\mathbb{C}^{2^t} $, with $\mathds1 \coloneqq \mathds1_1$, and for any $A \in \mathbb{C}^{2 \times 2}$, we define $A_t^{(n)} \coloneqq \mathds1_{n-t} \otimes A \otimes \mathds1_{t-1}. $ For a set of matrices $\{ A^{(1)}, \dots, A^{(p)}\}$, we define $\prod_{i=1}^p A^{(i)} \coloneqq A^{(1)} \dots A^{(p)}$. For any complex matrix $A = \sum_{ij} A_{ij} \ket{i} \bra{j}$, $\| A \|_1 \coloneqq \sum_{ij} | A_{ij}|$ and $\| A\|_{\text{tr}} \coloneqq \text{tr}(\sqrt{A^{\dag} A})$. Any matrix $A \in \mathbb{C}^{2^t \times 2^t}$ such that $A = A^{(1)} \otimes \dots \otimes A^{(t)}$ for some $A^{(1)}, \dots, A^{(n)} \in \mathbb{C}^{2 \times 2}$ is called a \textit{product} matrix. 
    
\subsection{Quantum Computing}
    A (pure) quantum \textit{state} is a (rank one) positive semidefinite operator $\sigma \in \mathbb{C}^{d \times d}$, such that $\text{tr}(\sigma) = 1$. In quantum computing, a \textit{qubit} is the quantum counterpart of a classical bit. An $n$-qubit system's state can be represented by a quantum state in $\mathbb{C}^{2^n \times 2^n}$.
    
    A quantum \textit{gate} operating on $n$ qubits is represented by a unitary operator $U \in \mathbb{C}^{2^n \times 2^n}$, which transforms the state of an $n$-qubit system from $\sigma$ to $\sigma_U \coloneqq U \sigma U^{\dag}$. In particular, the one-qubit \textit{Pauli gates} are defined as 
    \begin{align}
            X \coloneqq 
            \begin{bmatrix}
                0 & 1 \\
                1 & 0
            \end{bmatrix}, \  
            Y \coloneqq 
            \begin{bmatrix}
                0 & -i \\
                i & 0
            \end{bmatrix},\ 
            Z \coloneqq 
            \begin{bmatrix}
                1 & 0 \\
                0 & -1
            \end{bmatrix}.
    \end{align}
    We define $\mathbb{P}_n $ as the set of all $n$-fold tensor products of the elements in $\{ \mathds{1}, X, Y, Z\}$. 
    
    A quantum \textit{circuit} is defined as a composition of multiple quantum gates. The \textit{width} of a quantum circuit is defined as the number of qubits on which it is defined.
\begin{figure*}[tbh] 
    \centering
    \begin{tabular}{cccc} 
         \includegraphics[width=0.4\columnwidth]{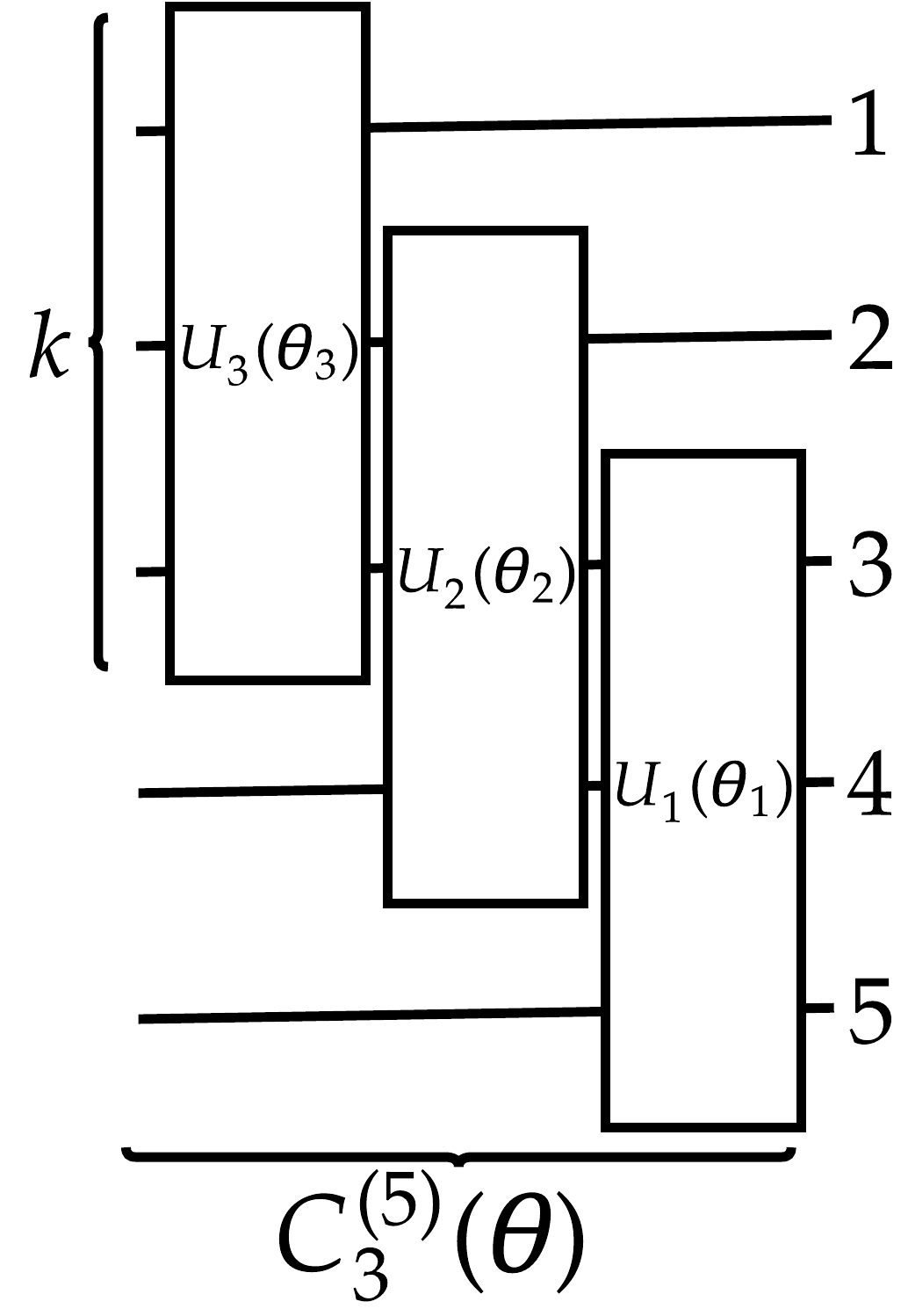} 
         &
        \includegraphics[width=0.39\columnwidth]{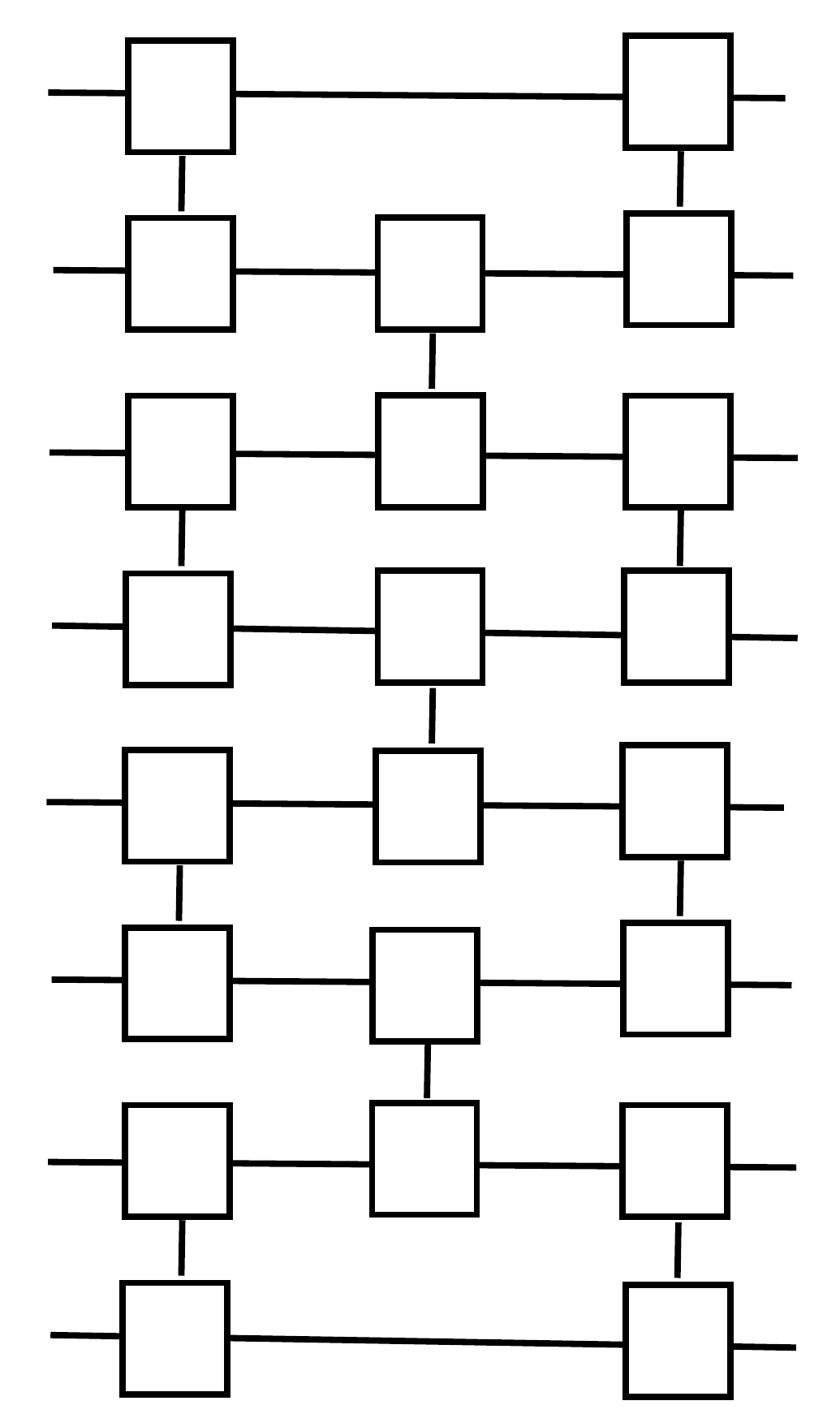}
        &
        \includegraphics[width=0.24\columnwidth]{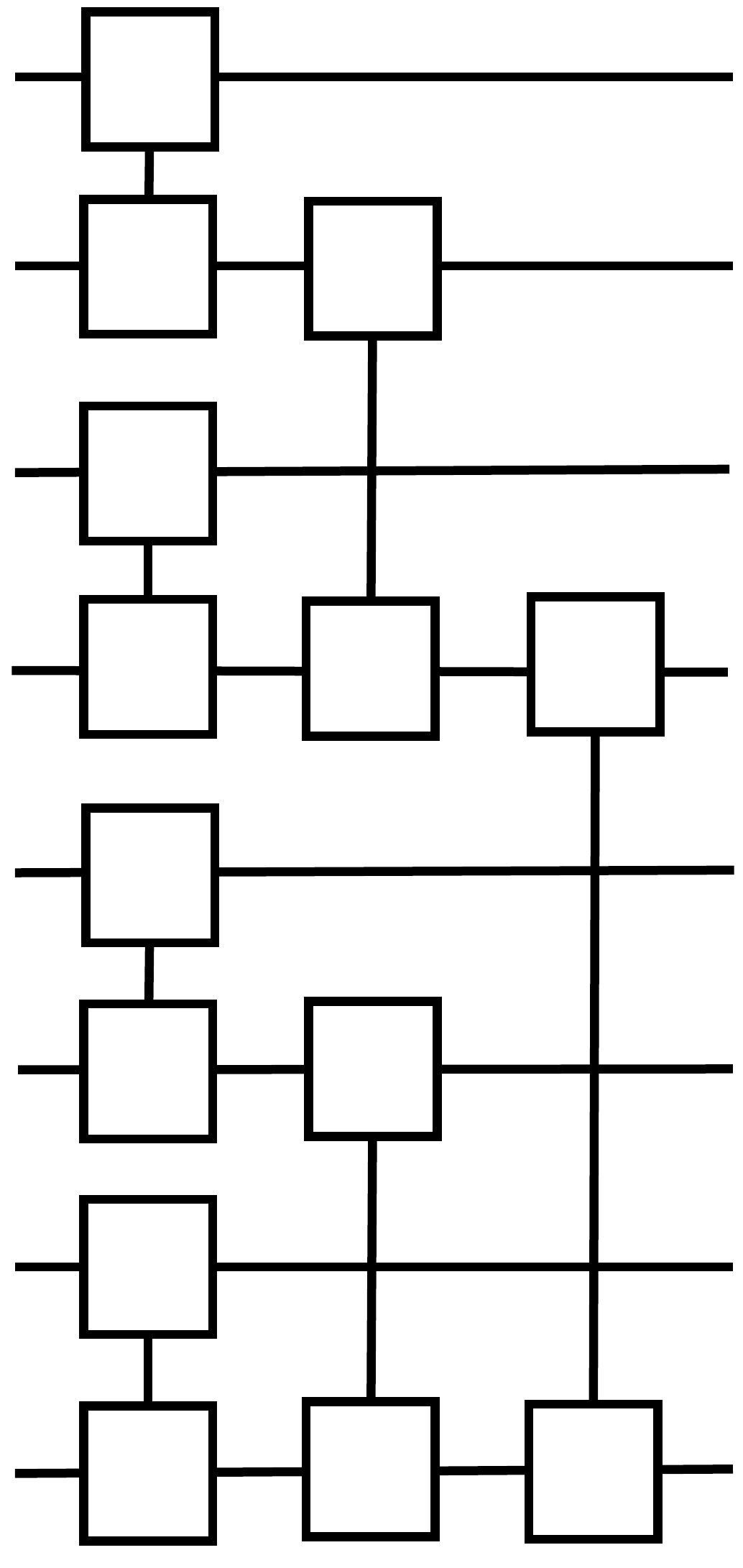}
        &
        \includegraphics[width=0.32\columnwidth]{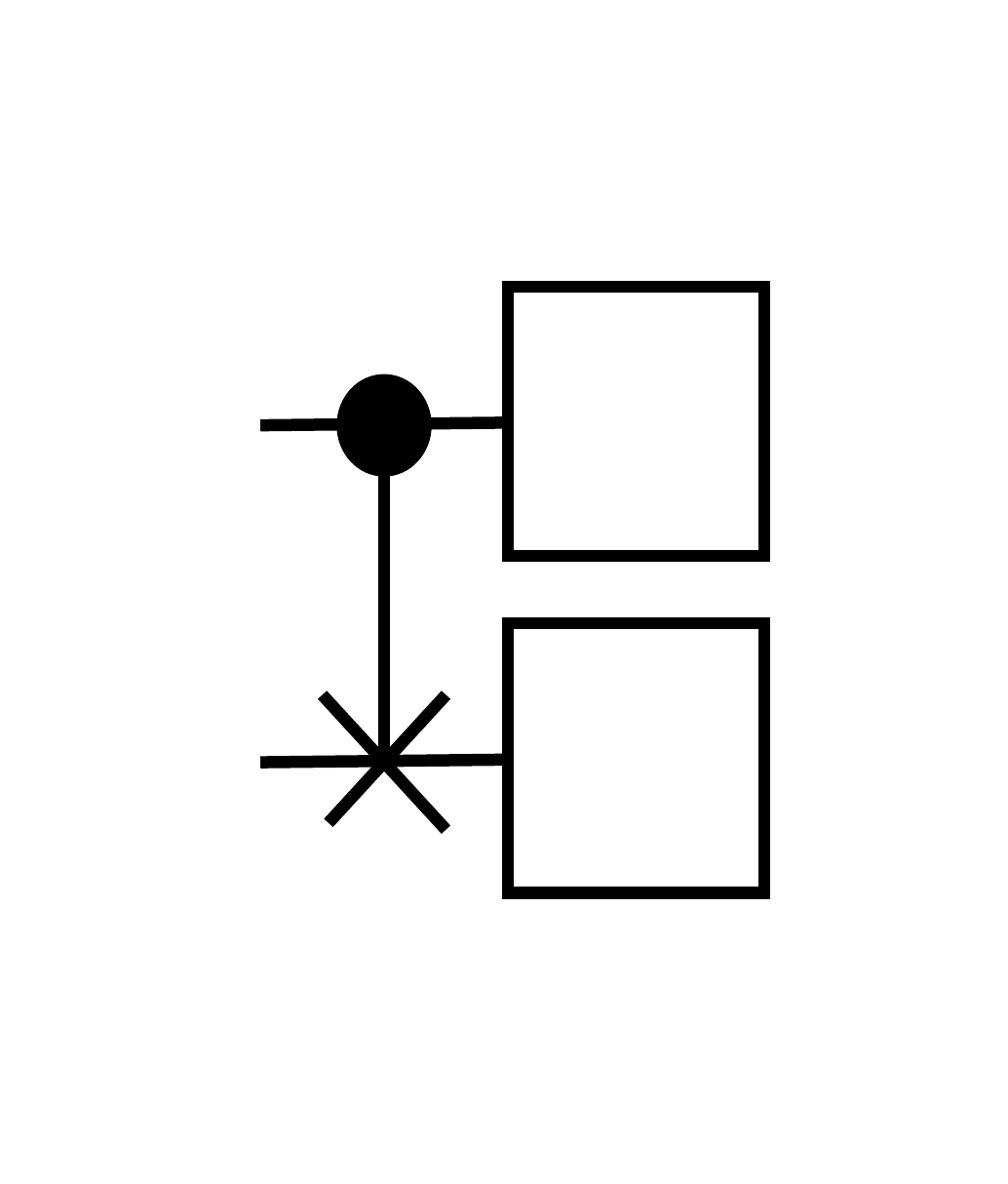}
        \\
        (a) & (b) & (c) & (d)
    \end{tabular} 
    \caption{(a) Example of an MPS ansatz with $n=5$ and $t = 3$. The numbers on the end of each qubit wire are the indices of the qubits. (b) HEA, where each connected pair of boxes are $2$-qubit subcircuits. (c) QCNN ansatz (d) The $2$-qubit subcircuit used in our simulations within HEA and QCNN, where each box is an arbitrary single qubit gate} \label{fig:ansatz}
    \end{figure*}
    A Hermitian operator, also known as an \textit{observable}, can be used to define a \textit{measurement} and an \textit{expectation value}. 
    A measurement of a system in a state $\sigma$, using an observable $W$ with eigendecomposition $W=\sum_i \omega_i \ket{w_i} \bra{w_i}$ results in an outcome $\omega_i$ with probability $ \bra{w_i} \sigma \ket{w_i}$. Thus, the expected value of the measurement outcome is $\text{tr}(W \sigma)$.
    
    We define $\mathbb{D}_n, \mathbb{U}_n$ and $\mathbb{H}_n$ as the set of all $n$-qubit pure states, gates, and observables. If a matrix $A$ acting on an $n$-qubit system acts as $\mathds1_{n-k}$ on $n-k$ qubits, then we call the matrix a $k$-local matrix. For any two pure $n$-qubit states $\sigma$ and $\rho$, the quantum \textit{fidelity}, defined as $F(\rho,\sigma) \coloneq \text{tr}(\rho \sigma)$, serves as a measure of similarity between the states with a fidelity of $1$ indicating that $\sigma = \rho$.

\subsection{Variational Quantum Algorithms}

Many optimization problems in quantum information can be formulated as the optimization of outputs from parameterized quantum circuits, also called \textit{ansatzes}, denoted as $C(\boldsymbol{\theta}) \in \mathbb{U}_n$, where $\boldsymbol{\theta}$ is a vector of parameters. Typically, these problems involve optimizing (over $\boldsymbol{\theta}$) the function defined as
\begin{align} \label{eq:main_optimization}
    f_{\sigma, W, C}(\boldsymbol{\theta}) = \text{tr}\left(WC(\boldsymbol{\theta}) \sigma C(\boldsymbol{\theta})^{\dag}\right) = \text{tr}\left(W\sigma_{C(\boldsymbol{\theta})} \right),
\end{align}
where $\sigma$ is the input quantum state, $\sigma_{C(\boldsymbol{\theta})}$ is $C(\boldsymbol{\theta})\sigma {C(\boldsymbol{\theta})}^\dagger$ by our notation,  and $W$ is an output observable. One can estimate this function for any input $\boldsymbol{\theta}$ by sample means estimation of the expected value of measuring the state $\sigma_{C(\boldsymbol{\theta})}$ using $W$. With this capability, all partial derivatives can also be estimated using standard methods such as finite differencing or quantum-specific ones such as the parameter shift rule~\cite{Mitarai2018}. The $\boldsymbol{\theta}$ search space has $ \mathcal{O}(\text{poly}(n))$ dimensions. Hence, $\boldsymbol{\theta}$ can be updated towards the optimum classically. Throughout this work, we omit the subscript $C$ and use the notation $f_{\sigma, W}$ since the choice of ansatz $C$ is usually implicitly understood from the context.

\subsection{MPS Ansatz} \label{sec:MPS_ansatz}
    
The MPS ansatz is given in Figure~\ref{fig:ansatz} (a). It is built using $ t$ cascading layers of smaller parameterized subcircuits $U_i(\boldsymbol{\theta}_i)$ for $i=1,\dots,t$, each with width $k$. Mathematically, the ansatz is defined as 
\begin{align*}
    C_t^{(n)}(\boldsymbol{\theta})
    = \prod \limits_{p=1}^t \mathds1 ^ {\otimes n-k-p+1} \otimes U_p(\boldsymbol{\theta}_p) \otimes \mathds1 ^ {\otimes p - 1}, 
\end{align*}
where $t \leq n-k+1$, $\boldsymbol{\theta} = \boldsymbol{\theta}_1 \oplus \dots \oplus \boldsymbol{\theta}_t$ with $ \boldsymbol{\theta}_p = [\theta_{p1}, \theta_{p2}, \dots, \theta_{pm}]$ and $U_p(\boldsymbol{\theta}_p) = \prod_{q=1}^{m} e^{-i\theta_{pq}H_{pq}}$ are $k$-qubit parameterized subcircuits, with $H_{pq} \in \mathbb{H}_k$.

This ansatz has a close relationship with the MPS data structure, which can store quantum states with space complexity that is positively correlated with the entanglement between neighboring qubits measured using their bond dimensions~\cite{Cirac2021}. From~\cite{Cramer2010}, we can see that every state that can be represented efficiently as an MPS with bond dimensions at most $2^{k-1}$ can be implemented using this ansatz (assuming that $U_p$ can implement any $k$-qubit unitary). This is what led many works to use the MPS ansatz to solve state approximation problems variationally~\cite{Lin2021,Rudolph2022,Dov2022,Ran2019,Rieser2023}.

Throughout this work, we set $T = n - k + 1$, and our focus is on $C_T^{(n)}$. Also, in appropriate contexts, we denote $C_T^{(n)}$ as $C_T$, as the dependency on $n$ is implied by the system's size. 

\subsection{State Approximation Using MPS Ansatz} \label{subsec:state_approximation}

The problem that we consider here is as follows: given many copies of an $n$-qubit pure state $\sigma$ (or stronger access to a circuit that can prepare $\sigma$), output the parameters of an MPS ansatz such that we can approximately prepare $\sigma$ using these parameters. We present two different approaches to solve this problem, differing only in their measurement strategies.

In this first method, the idea is to maximize the fidelity between $  \ket{\mathbf{0}}\bra{\mathbf{0}}_{{C_T (\boldsymbol{\theta})} ^ {\dag}} $ and $\sigma$, over $\boldsymbol{\theta}$. That is, find $\argmax_{\boldsymbol{\theta}} f_{\sigma,\ket{\mathbf{0}}\bra{\mathbf{0}}}(\boldsymbol{\theta})$. This fidelity can be estimated by applying $ C_T (\boldsymbol{\theta})$ on $\sigma$, measuring all qubits simultaneously using the observable $Z$, and estimating the probability of all measurements resulting in $+1$. So the observable whose expectation features within the objective function is the \textit{global observable} $\ket{\mathbf{0}} \bra{\mathbf{0}}$. 

Alternatively, in the second method, we employ the expectation of $O \coloneqq 1/n\sum_{i=1}^n \ket{0} \bra{0}_i $, which is a sum of $1$-local observables. The intuition here is that if $\boldsymbol{\theta}^*$ maximizes $f_{\sigma,\ket{\mathbf{0}}\bra{\mathbf{0}}}$, then $ \sigma_{C_T(\boldsymbol{\theta}^*)} = \ket{\mathbf{0}}\bra{\mathbf{0}}$ and hence $ f_{\sigma, O}$ also attains its maximum on $\boldsymbol{\theta}^*$. Moreover, for any $\boldsymbol{\theta}$, within $f_{\sigma, O}(\boldsymbol{\theta}) = 1/n \sum_{i=1}^n f_{\sigma, \ket{0} \bra{0}_i}(\boldsymbol{\theta}) $, the $i^{\text{th}}$ summand can be estimated by measuring the $i^{\text{th}}$ qubit of $\sigma_{C_T(\boldsymbol{\theta})}$ using $ Z$ and estimating the probability that it yields $ +1$.  

\section{Trainability}
    In this section, we present our theoretical results regarding cost concentration and barren plateaus of state approximation carried out using global and local observables. 

\subsection{Cost Concentration} \label{sec:cost_concentration}

We start with a formal definition of cost concentration.

\begin{definition}
    Let $\sigma \in \mathbb{D}_n$, $W \in \mathbb{H}_n$, and $\boldsymbol{\theta}$ be distributed over a compact parameter space. For any ansatz $C$, $f_{\sigma,W}$ exhibits cost concentration if 
    \begin{align}
        \text{Var}_{\boldsymbol{\theta}} \left( f_{\sigma, W}(\boldsymbol{\theta})\right) \in \mathcal{O}\left(\frac{1}{b^n} \right)
    \end{align}
    for some $b>1$.
\end{definition}

From the above definition, we can see that for any VQA objective function that exhibits cost concentration, the output would be exponentially concentrated around $ \mathbb{E}_{\boldsymbol{\theta}} \left( f_{\sigma, W}(\boldsymbol{\theta}) \right)$. Thus, if there exists a space $\mathbb{V}$ such that $\forall \ \boldsymbol{\theta} \in \mathbb{V}$, $f_{\sigma, W}(\boldsymbol{\theta}) \geq \mathbb{E}_{\boldsymbol{\theta}} \left( f_{\sigma, W}(\boldsymbol{\theta}) \right) + \Omega(1)$, then, $\mathbb{V}$ must have $\mathcal{O}(1/b^n)$ measure. Further, cost concentration ensures that the landscape of $f_{\sigma, W}(\boldsymbol{\theta})$ is extremely flat for almost all $\boldsymbol{\theta}$. Thus gradient variations are too small to distinguish from zero using at most $\mathcal{O}(\text{poly}(n))$ samples ensuring there is no strong gradient to follow to reach $\mathbb{V}$. Therefore, to estimate these outputs with meaningful precision, one would require an exponentially large number of samples, or equivalently, copies of $\sigma$.

Now, we present our theoretical results regarding cost concentration in learning MPS approximations variationally using $C_T$. Many trainability results in the literature assume one of two assumptions on the input state~\cite{Cerezo2024}; either they are "close" to product states~\cite{Pesah2021,Cerezo2021} or they are sparse~\cite{Monbroussou2023,Larocca2022,Cherrat2023}. Our results also make such assumptions and hence use $h_1(\sigma) \coloneqq \min_{V_1, \dots, V_n \in \mathbb{U}_2} \| \sigma_{V_1 \otimes \dots \otimes V_n}\|_1^2$ and $h_2(\sigma) = \min_{\rho_1, \dots, \rho_n \in \mathbb{D}_1} \| \rho_1 \otimes \dots \otimes \rho_n - \sigma\|_{\text{tr}}$, to quantify sparsity and proximity to product states respectively.

We start by proving that using global observables for state approximation can give rise to an objective function that exhibits cost concentration.
\begin{theorem} \label{th:upper_bound}
    Let $\sigma \in \mathbb{D}_n$ and $C_T^{(n)}$ be an MPS ansatz where each parameterized subcircuit $U_i$ forms a unitary $2$-design. Then, we have
    \begin{align}
        \text{Var}_{\boldsymbol{\theta}} \left(  f_{\sigma,\ket{\mathbf{0}}\bra{\mathbf{0}}}(\boldsymbol{\theta}) \right) \leq \frac{h_1(\sigma)}{4^ {n-k-1}}.
    \end{align}
\end{theorem}
Hence, for states with $h_1(\sigma) \in \mathcal{O}(4^{n/p})$ with $p > 1$, we see that the upper bound will decrease exponentially. The proof of Theorem~\ref{th:upper_bound}, as well as all other theorems introduced in this work, can be found in the Appendix. 

In contrast,  the next theorem shows that the alternative method leveraging local observables provably avoids cost concentration. 
\begin{theorem} \label{th:lower_bound}
    Let $\sigma \in \mathbb{D}_n$, $O \coloneqq 1/n\sum_{i=1}^n \ket{0} \bra{0}_i$, and $C_T^{(n)}$ be an MPS ansatz, where each parameterized subcircuit $U_i$ forms a unitary $2$-design. Then, we have 
    \begin{align}
        \text{Var}_{\boldsymbol{\theta}} \left(f_{\sigma, O}(\boldsymbol{\theta}) \right)  \geq \frac{1}{n(2 ^ {2k+1} + 4)} - \frac{h_2(\sigma)}{2n}.
    \end{align}
\end{theorem}

We note that for any product state $\sigma$, $h_2(\sigma)=0$. More generally, when $h_2(\sigma) \ll  1/(2 ^ {2k} + 2) $, the lower bound scales linearly in $n$ and exponentially only in $k$.

The core idea behind both proofs is to integrate each $ U_t$ starting from $U_T$ using standard Haar random integration methods (cf. Lemma~\ref{le:2_design} in Appendix). Typically, this would yield a linear combination of multiple terms, each being an expectation of MPS ansatz circuit outputs with the same observables but defined over $n-T+t-1$ qubit systems and different input states that were \textit{dependent} on the previous state. Thus, naively integrating each $U_t$ one at a time requires integrating a number of terms exponential in $T$. However, we demonstrate that for the MPS ansatz and the state classes in Theorems~\ref{th:upper_bound} and~\ref{th:lower_bound}, integrating any $U_t$ results in a linear combination of such terms that are \textit{independent} of the previous state, with such state dependency only in the coefficients. This allows us to compute all $T $ integrations using products of $T $ matrices, whose dimension is the number of terms in the linear combination, which in our case, is 2.

Our experimental results discussed later in this work used input states with $h_1$ and $h_2$ not necessarily small, suggesting this or similar bounds for a wider variety of states may hold.

Some works in the literature that use the MPS ansatz consider efficient MPS descriptions of states as input, rather than actual quantum states. In such cases, the entire VQA optimization can be efficiently implemented on classical computers using tensor network simulation~\cite{Jozsa2006}. Within such methods, the objective function is evaluated exactly, not estimated, so cost concentration is not an issue. However, as we will see in the next section, cost concentration also leads to barren plateaus, which can cause parameter updates to be exponentially small, thus hindering even fully classical optimization protocols.

   \subsection{From Cost Concentration to Barren Plateaus} \label{sec:barren_plateaus}
    In this section, we discuss the relationship of Theorems~\ref{th:upper_bound} and~\ref{th:lower_bound} to barren plateaus. First, we formally define barren plateaus, as per~\cite{Arrasmith2022}.

    \begin{definition}
        Let $\sigma \in \mathbb{D}_n$ and let $W \in \mathbb{H}_n$. For any ansatz $C(\boldsymbol{\theta}) = \prod \limits_{p=1}^t U_p(\boldsymbol{\theta}_p) $, where $ U_p(\boldsymbol{\theta}_p) = \prod \limits_{q=1}^m e ^ {-i \theta_{pq} H_{pq}}$, $\boldsymbol{\theta}_p = [\theta_{p1} \dots \theta_{pm}]$, $H_{pq} \in \mathbb{H}_n$ and $ \boldsymbol{\theta} = \boldsymbol{\theta}_1 \oplus \dots \oplus \boldsymbol{\theta}_t$, and for any $p,q$, define
        \begin{align}
            U_p^{(L,q)} \left(\boldsymbol{\theta}_p \right) &= \prod \limits_{j=1}^{q-1} e^{-i \theta_{pj} H_{pj}}, \\
            U_p^{(R,q)} \left(\boldsymbol{\theta}_p \right) &= \prod \limits_{j=q+1}^{m} e^{-i \theta_{pj} H_{pj}}.
        \end{align}
        Then, $ f_{\sigma, W}$ exhibits a barren plateau if $\forall \ p,q$ satisfying $ 1 \leq p \leq t, 1 \leq q \leq m$, we have
        \begin{align}
            \text{Var}_{\boldsymbol{\theta}} \left( \partial_{\theta_{pq}} f_{\sigma, W}(\boldsymbol{\theta})\right) \in \mathcal{O}\left(\frac{1}{b^n}\right),
        \end{align}
        for some constant $b>1$, where $ \partial_{\theta_{pq}} f_{\sigma, W}(\boldsymbol{\theta})$ is it's partial derivative with respect to $\theta_{pq}$ and $U_1, \dots U_{p-1}, U_{p+1}, \dots U_t$, along with one of $U_p^{(L,q)}$ or $U_p^{(R,q)}$ are distributed according to the Haar measure and $\theta_{pq}$ is distributed uniformly.  
    \end{definition}
     Just as cost concentration makes the outputs of most inputs exponentially concentrated, barren plateaus cause most partial derivatives to be exponentially small, since $\mathbb{E}_{\boldsymbol{\theta}} \left( \partial_{\theta_{pq}}\text{tr}(W\sigma_{C(\boldsymbol{\theta})}) \right) = 0 \ \forall \ p,q$~\cite{Cerezo2021}. This means that estimating these derivatives will require exponential resources, and in most cases, the parameter updates that gradient-based classical optimizers make will be exponentially small. 
     
    Next, we will use Theorem~\ref{th:upper_bound} to demonstrate that employing the MPS ansatz for learning state approximations leads to barren plateaus when global observables are used.


    \begin{corollary} \label{co:upper_bound}
        Let $\sigma \in \mathbb{D}_n$ and $C_T^{(n)}$ be an MPS ansatz. Then, $f_{\sigma, \ket{\mathbf{0}} \bra{\mathbf{0}}}$ exhibitis barren plateaus if $ h_1(\sigma) \in \mathcal{O}(4^{n/p})$.
    \end{corollary}

    Similarly, we extend Theorem~\ref{th:lower_bound} to demonstrate that using the MPS ansatz with local observables prevents barren plateaus.

    \begin{corollary} \label{co:lower_bound}
        Let $\sigma \in \mathbb{D}_n$, $C_T^{(n)}$ be an MPS ansatz, and $O \coloneqq 1/n\sum \limits_{i=1}^n \ket{0}\bra{0}_i$. Then, $f_{\sigma, O}$ exhibitis barren plateaus if $ h_2(\sigma) \ll  1/(2 ^ {2k} + 2)$.
    \end{corollary}
    
\section{Towards Classical Simulation Through Effective Subspaces} \label{sec:classical_simulation}

    In this section, we discuss the possibility of designing an efficient classical algorithm capable of simulating state approximation VQAs involving MPS ansatzes and local observables, using very few copies of the input quantum state. 
    
    The idea builds on the conjecture from~\cite{Cerezo2024} which says that any objective function avoiding cost concentration exhibits effective subspaces, a property useful for designing classical simulation algorithms with minimal quantum resources. Our simulations demonstrate that objective functions involving MPS ansatz and local observables, which we previously proved to avoid cost concentration, indeed exhibit effective subspaces within the Pauli basis, further supporting this conjecture.

    Note that in this work, we do not present an explicit algorithm for the aforementioned classical simulation, but rather present evidence that such a protocol could exist. First, we introduce effective subspaces as outlined in~\cite{Cerezo2024}.
    
    \subsection{Effective Subspace}
    
    Let $C(\boldsymbol{\theta})$ be an $n$-qubit ansatz and let $W \in \mathbb{H}_n, \sigma \in \mathbb{D}_n$. Effective subspaces are loosely defined as follows:

    \begin{definition} \cite{Cerezo2024}
        For any orthonormal basis $ \mathbb{K} = \{ K_1, K_2, \dots K_{4^n}\}$ of $\mathbb{C}^{2^n \times 2^n}$, and for any $\boldsymbol{\theta}$, define a distribution $\mathcal{P}_{\boldsymbol{\theta}, W, \mathbb{K}}$ over $\mathbb{K}$ as 
        \begin{align}
        \mathcal{P}_{\boldsymbol{\theta}, W,\mathbb{K}}(K_j) = \frac{ f_{K_j,W}(\boldsymbol{\theta}) ^ 2}{\| W\|_2^2}.
        \end{align}
        An ansatz and observable combination exhibits an effective subspace if there exists a basis $\mathbb{K}$ such that for almost all $\boldsymbol{\theta}$, $\mathcal{P}_{\boldsymbol{\theta}, W,\mathbb{K}}(K_j)$ is large only for $K_j$ in a subset $ \mathbb{K}_s \subset \mathbb{K}$, that is independent of $\boldsymbol{\theta}$ and has $| \mathbb{K}_s| \in \mathcal{O} (\text{poly(n)})$.
    \end{definition}  
    Note that the elements of $\mathbb{K}$ are not restricted to quantum states. Also, in appropriate contexts, we denote $ \mathcal{P}_{\boldsymbol{\theta}, W,\mathbb{K}} (K_j)$ as $\mathcal{P}_{\boldsymbol{\theta}, W}(K_j)$ as the dependency on $\mathbb{K}$ is implicitly understood from the context.
    
        In~\cite{Cerezo2024}, it is conjectured, with evidence, that all ansatz-observable combinations that have been shown to provably avoid barren plateaus exhibit an effective subspace, at least for some subset of input states. Popular examples involving shallow ($\mathcal{O}(\log n)$-depth) ansatzes include HEA-local observable and the QCNN-local observable combinations. In both these cases, the basis $\mathbb{K}$ can be $\mathbb{P}_n$. The presence of effective subspaces means that if we estimate $ \text{tr}(K\sigma)\ \forall \ K \in \mathbb{K}_s$ as a preprocessing step, and if we can classically compute $ \text{tr}(EW_{C(\boldsymbol{\theta})^{\dag}}) \ \forall \ K \in \mathbb{K}_s \ \text{and} \ \forall \  \boldsymbol{\theta}$ efficiently, then in many cases, $ f_{\sigma, W}(\boldsymbol{\theta})$ can be classically estimated with good precision, because   
        \begin{align*}
            f_{\sigma,W}(\boldsymbol{\theta}) = \text{tr}\left(W\sigma_{C(\boldsymbol{\theta})} \right) &= \text{tr}\left(W_{C(\boldsymbol{\theta})^{\dag}}\sigma \right) \\
            &= \sum \limits_{K \in \mathbb{K}} \text{tr}\left(KW_{C(\boldsymbol{\theta})^{\dag}}\right) \text{tr}(K \sigma),
        \end{align*}
        and if most $\text{tr}\left(KW_{C(\boldsymbol{\theta})^{\dag}}\right)$ is large only for those $K \in \mathbb{K}_s$, then
        \begin{align}
            f_{\sigma,W}(\boldsymbol{\theta}) \approx \sum \limits_{K \in \mathbb{K}_s} \text{tr}\left(KW_{C(\boldsymbol{\theta})^{\dag}}\right) \text{tr}(K \sigma).
        \end{align}
        This is the underlying principle behind designing classical simulations using effective subspaces.  

        When it comes to the classical simulation of $ f_{\sigma, O, C_T}(\boldsymbol{\theta}) = 1/n\sum_i f_{\sigma, \ket{0}\bra{0}_i, C_T}(\boldsymbol{\theta})$, it is sufficient to be able to classically estimate each $ f_{\sigma, \ket{0}\bra{0}_i}(\boldsymbol{\theta})$ efficiently. From Figure~\ref{fig:ansatz} (a), it is easy to see that for any product state $\sigma$, among these $n$ terms, the hardest to estimate are $f_{\sigma, \ket{0}\bra{0}_i}(\boldsymbol{\theta})$ for $i \in \{n-k+1, \dots, n\}$, because for all other $i$, at least one subcircuit within $C_T(\boldsymbol{\theta})$ will be canceled leaving an expression of the same form on fewer qubits. When $q$ subcircuits are canceled, at least $4^{n-q}(4^q-1)$ outcomes of $\mathcal{P}_{\boldsymbol{\theta}, \ket{0}\bra{0}_i,\mathbb{P}_n}$ will be zero for any $\boldsymbol{\theta}$, making the distribution very concentrated. Using Lemma~\ref{le:1_design_multi} in the Appendix, we find that for any $i,j \in \{n-k+1, \dots, n\}$, $\mathbb{E}_{\boldsymbol{\theta}} ]( f_{\sigma, \ket{0}\bra{0}_i}(\boldsymbol{\theta}) ) = \mathbb{E}_{\boldsymbol{\theta}} ( f_{\sigma, \ket{0}\bra{0}_j}(\boldsymbol{\theta}))$. Therefore, we focus on $f_{\sigma, \ket{0}\bra{0}_n}(\boldsymbol{\theta})$ and aim to show that the $C_T$-$\ket{0}\bra{0}_n$ combination also exhibits an effective subspace with $\mathbb{K} = \mathbb{P}_n$.

\begin{figure}[tbh]
    \centering
    \begin{tabular}{c}
        \includegraphics[width=0.9\columnwidth]{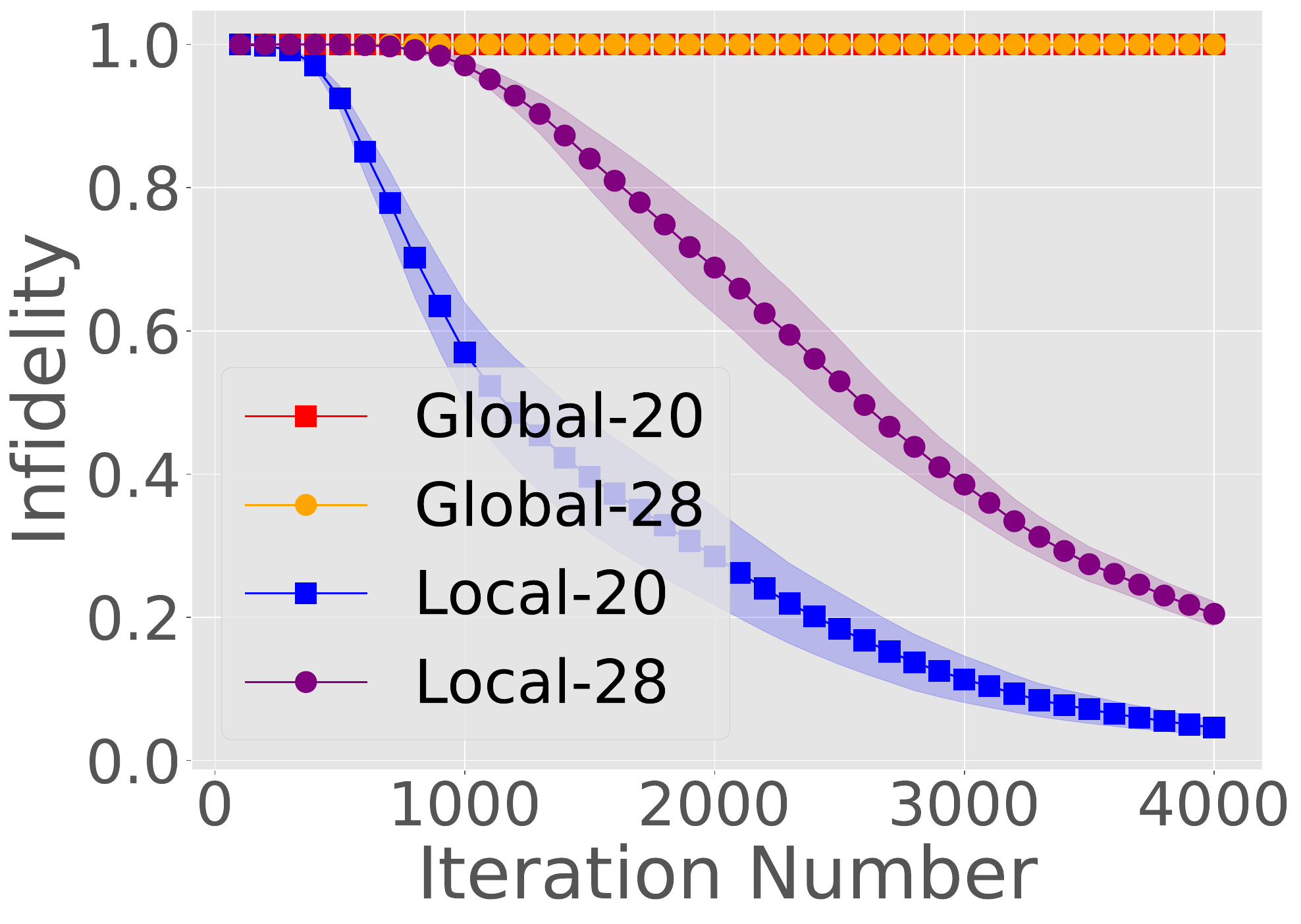} \\
        (a) \\
         \includegraphics[width=0.9\columnwidth]{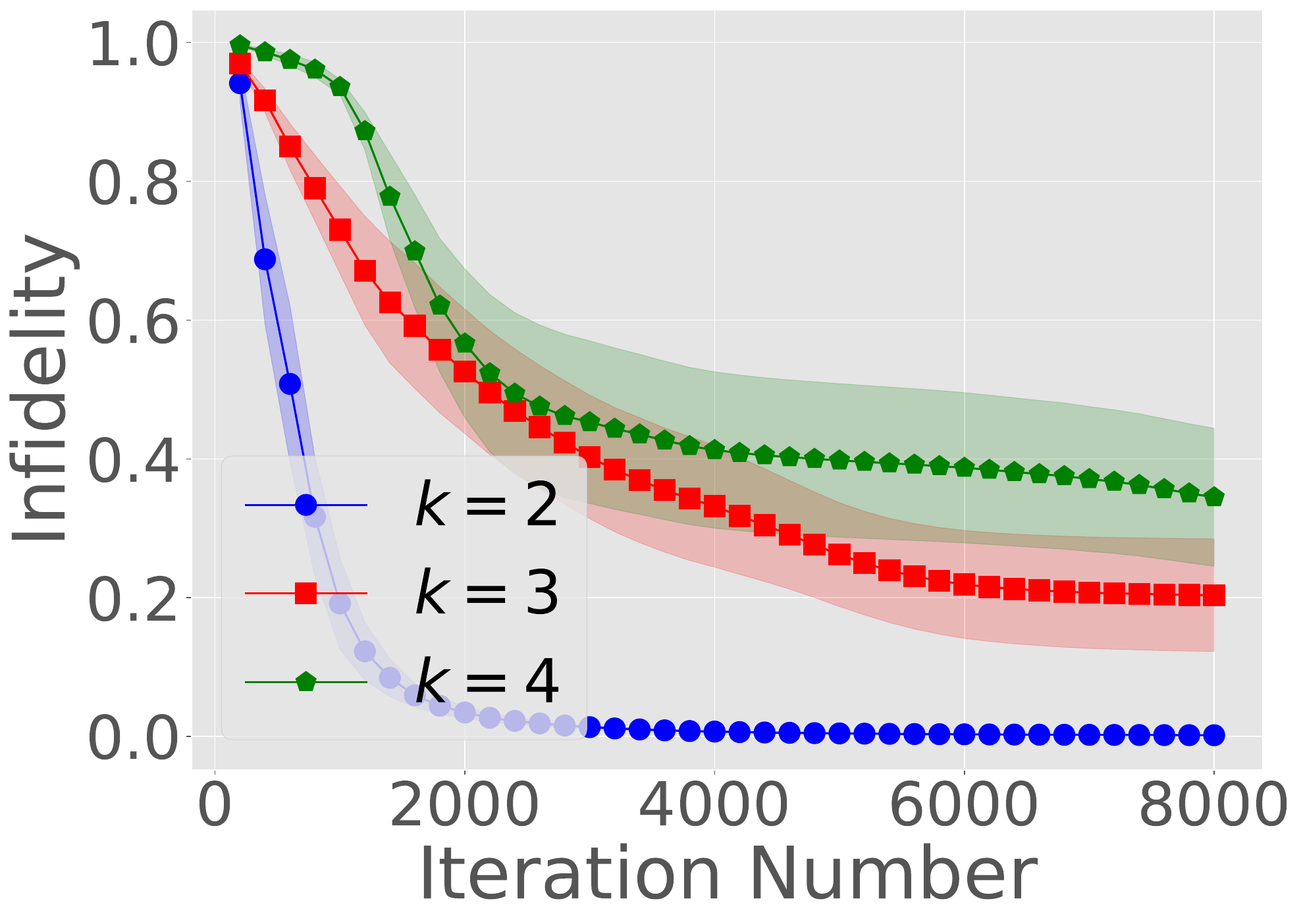}
         \\
         (b)
    \end{tabular}
    \caption{Simulation results of state approximation using MPS ansatz. In (a), we plot the learning curves of state approximation using the MPS ansatz with subcircuit width $2$ for the target state $\ket{\mathbf{0}} \bra{\mathbf{0}}$, optimized by SPSA with $n=20$ and $n=28$. The results demonstrate that global observables significantly hinder the learning process. In (b) state approximation results for the same target state using the MPS ansatz with local observables are plotted for $n=12$ with varying subcircuit widths $k$. The plots indicate that increasing the subcircuit width progressively impairs learning efficiency.} \label{fig:state_approximation}
    \end{figure} 

        \subsection{C-$\mathbb{K}$ Norm} \label{subsec:ce_norm}
           Now, we introduce a norm that can be used to measure how concentrated the distributions $ \mathcal{P}_{\boldsymbol{\theta}, W}$ would be, for typical values of $\boldsymbol{\theta}$. Given any discrete distribution $\mathcal{P}$, $\| \mathcal{P}\|_2$ can be used to measure how concentrated the distribution is. A higher $\| \mathcal{P}\|_2$ indicates that the distribution is concentrated among a few outcomes with high probability. Hence, we can use the $2$-norm of the distributions $\mathcal{P}_{\boldsymbol{\theta}, W}$ to assess how concentrated these distributions are. So, we define the $\mathbb{K}$-norm (in $\mathbb{H}_n$) as this $2$-norm, that is
        \begin{align}
            \| W\|_{\mathbb{K}} \coloneqq \frac{1}{\| W\|_2}\left[\sum \limits_{K \in \mathbb{K}} {\text{tr}(KW)} ^ 4 \right]^{1/4}
        \end{align}
        Note that for any $\boldsymbol{\theta}$, $\| W_{{C(\boldsymbol{\theta})}^{\dag}}\|_{\mathbb{K}}$ is simply the $2$-norm of the distribution $\mathcal{P}_{\boldsymbol{\theta}, W}$. We first prove the following result regarding the cost of computing $\left\Vert W_{{C(\boldsymbol{\theta})}^{\dag}} \right\Vert_{\mathbb{K}}$ for any $\boldsymbol{\theta}$.
        \begin{theorem} \label{th:classical_cost}
            For any $n$-qubit quantum circuit $V$, where the qubits are arranged in a line, let $R_V = \max_i R_{V,i}$, where $R_{V,i}$ is the number of 2-qubit gates being applied on any qubits $j, k$ such that $j \leq i \leq k$. Then, for any product observable $W$, $ \left\Vert W_{V^{\dag}}\right\Vert_{\mathbb{K}}$ can be classically evaluated with cost $\mathcal{O}\left(2^{R_V}\right)$.
        \end{theorem}
        
        Although the MPS ansatz is defined using $k$-qubit parameterized gates, Theorem~\ref{th:classical_cost} concerns only $2$-qubit gates because, in practice, we always decompose each $k$-qubit unitary into smaller $2$-qubit parameterized gates. As an example, in all our simulations, all $k$-qubit unitaries are HEAs (cf. Figure~\ref{fig:ansatz} (b)).

        From Figure~\ref{fig:ansatz} (a), we can see that $R_V$ is independent of $n$. Typically, it scales as $\mathcal{O}(\text{poly}(k))$ meaning that the cost of evaluating $ \left\Vert W_{{C_T(\boldsymbol{\theta})}^{\dag}}\right\Vert_{\mathbb{K}}$ will be $\mathcal{O}\left(2 ^ {\text{poly}(k)}\right)$. 

        Now, as mentioned earlier, we would like to analyze $ \left\Vert {\ket{0}\bra{0}_n}_{{C_T(\boldsymbol{\theta})}^{\dag}}\right\Vert_{\mathbb{K}}$ for typical values of $ \boldsymbol{\theta}$. Hence, we introduce the $C$-$\mathbb{K}$ norm in the following theorem. 

        \begin{theorem}
            For any parameterized circuit $C$, and an orthonormal basis $\mathbb{K}$ of $\mathbb{C}^{2^n \times 2^n}$, define 
            \begin{align} \label{eq:C_norm}
                \| W\|_{C,\mathbb{K}} \coloneqq \int \limits_{\boldsymbol{\theta}} \left\Vert W_{{C(\boldsymbol{\theta})}^{\dag}}\right\Vert_{\mathbb{K}} \text{d} \boldsymbol{\theta},
            \end{align}
            for any $W \in \mathbb{H}_n$. Then, $\| \cdot \|_{C,\mathbb{K}}$ is a norm on $\mathbb{H}_n$.
        \end{theorem}
\begin{figure*}[tbh]
    \centering
    \begin{tabular}{ccc}
        MPS & HEA & QCNN \\
        \includegraphics[width=0.65\columnwidth]{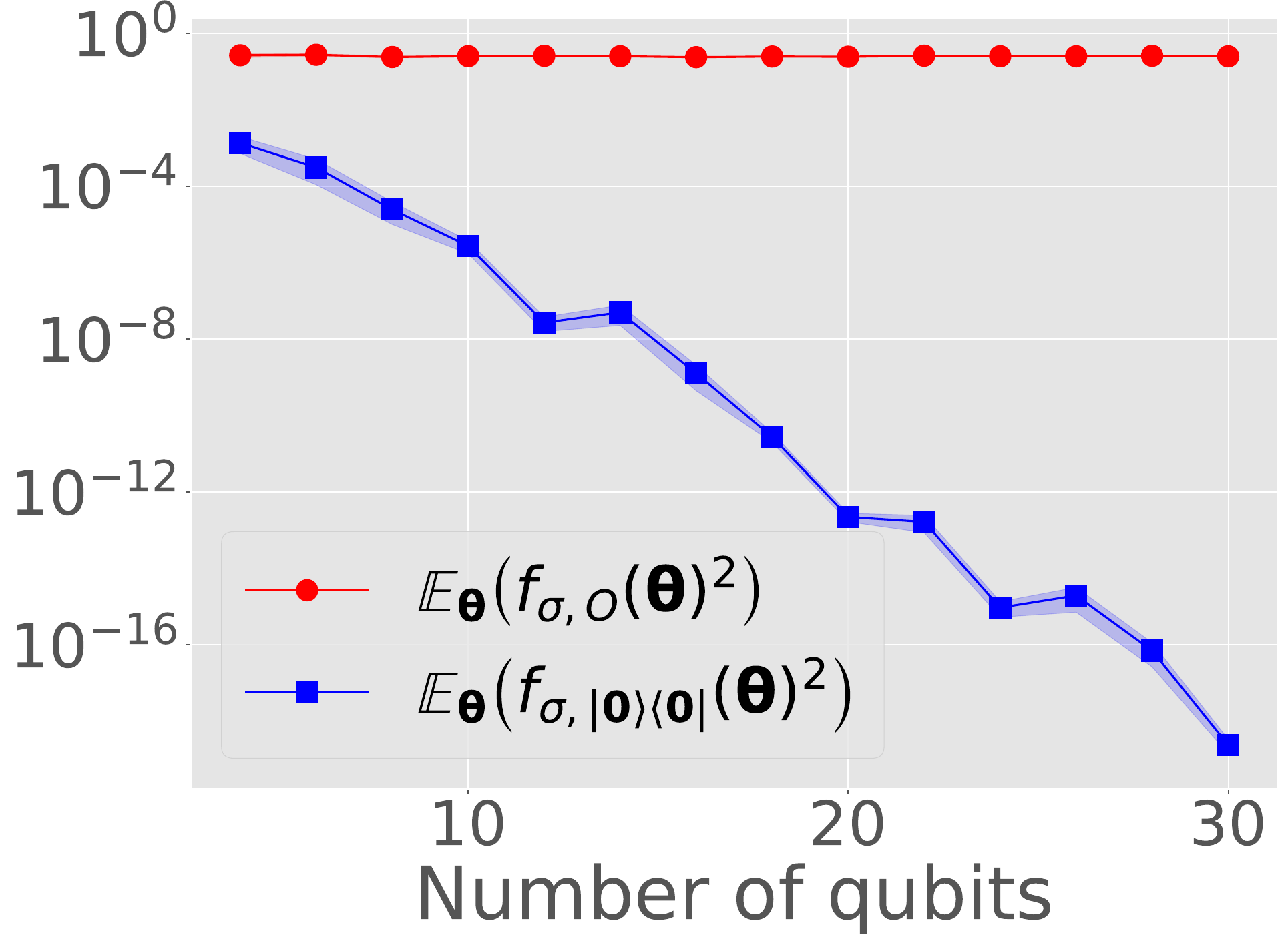} 
         &
        \includegraphics[width=0.65\columnwidth]{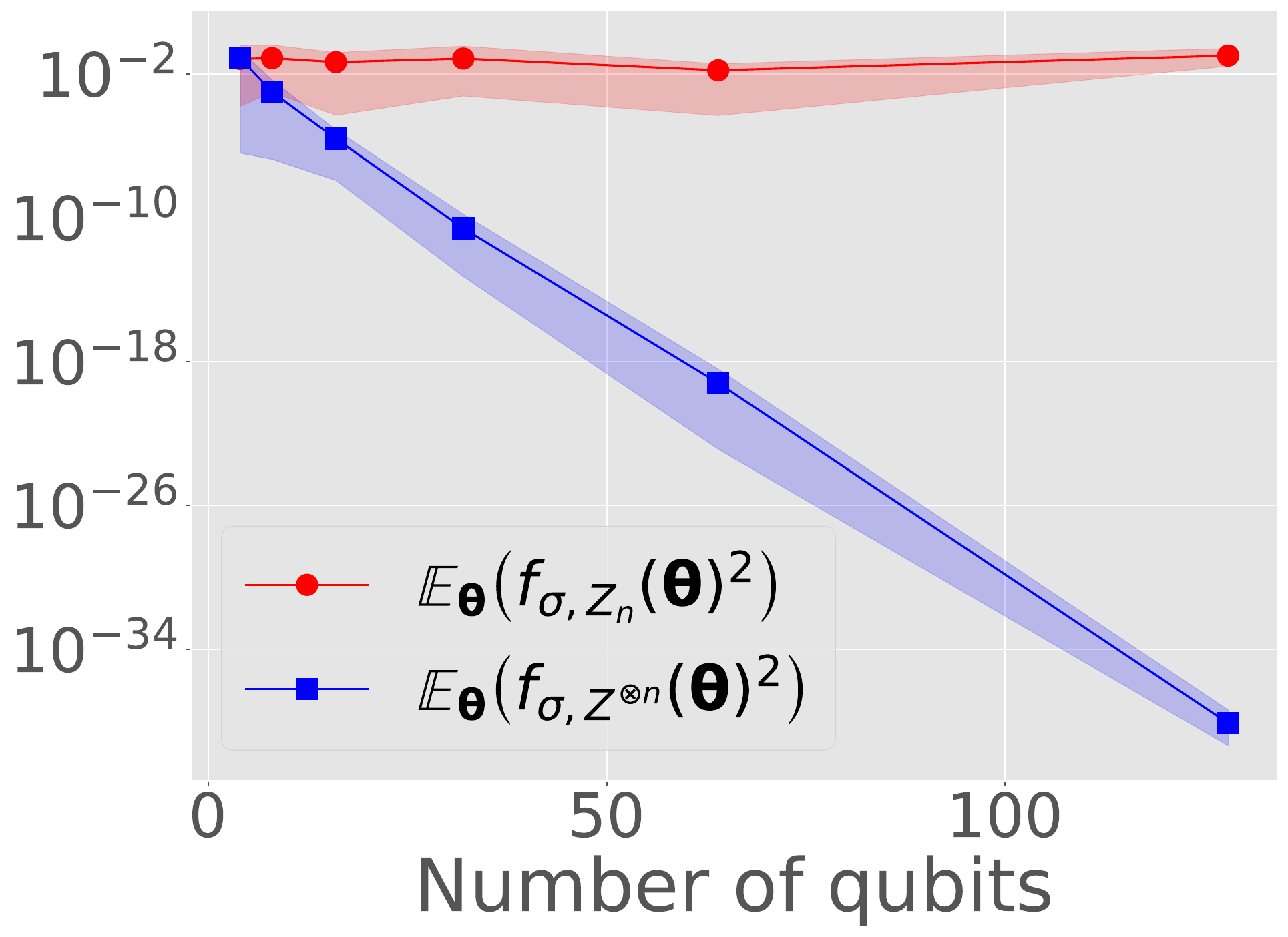}
        &
        \includegraphics[width=0.65\columnwidth]{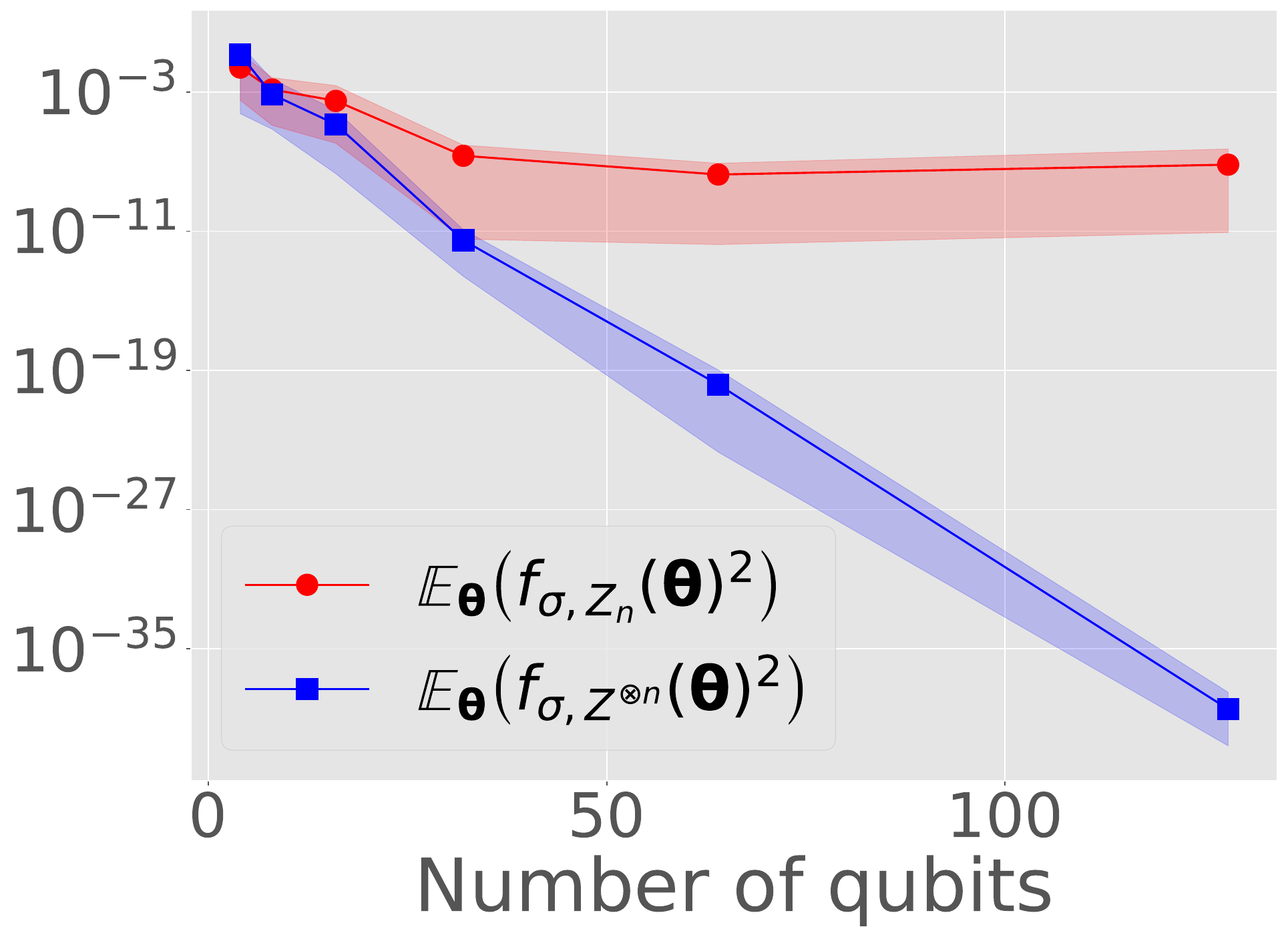} \\
        (a) & (b) & (c) \\
        \includegraphics[width=0.65\columnwidth]{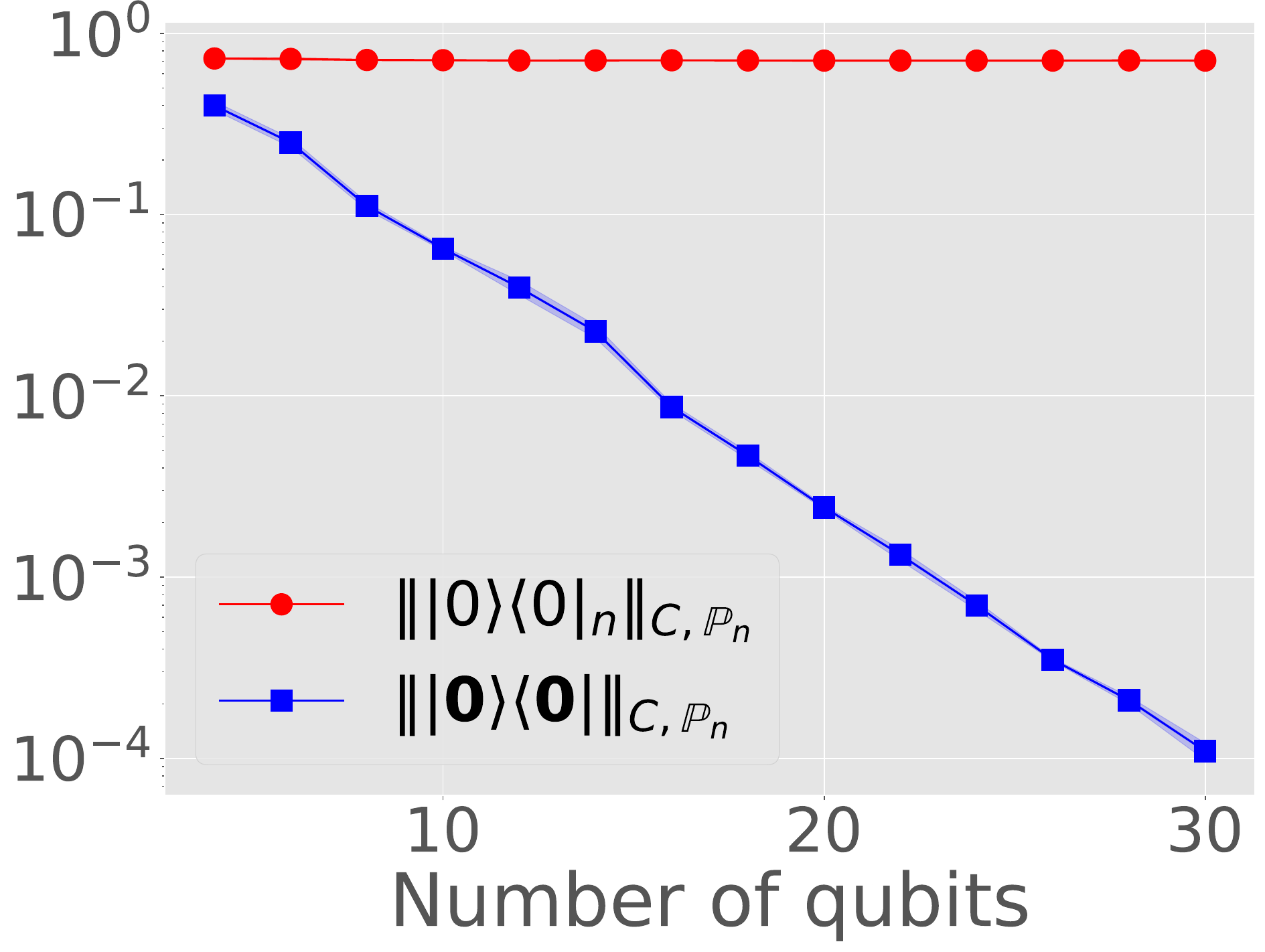} 
         &
         \includegraphics[width=0.65\columnwidth]{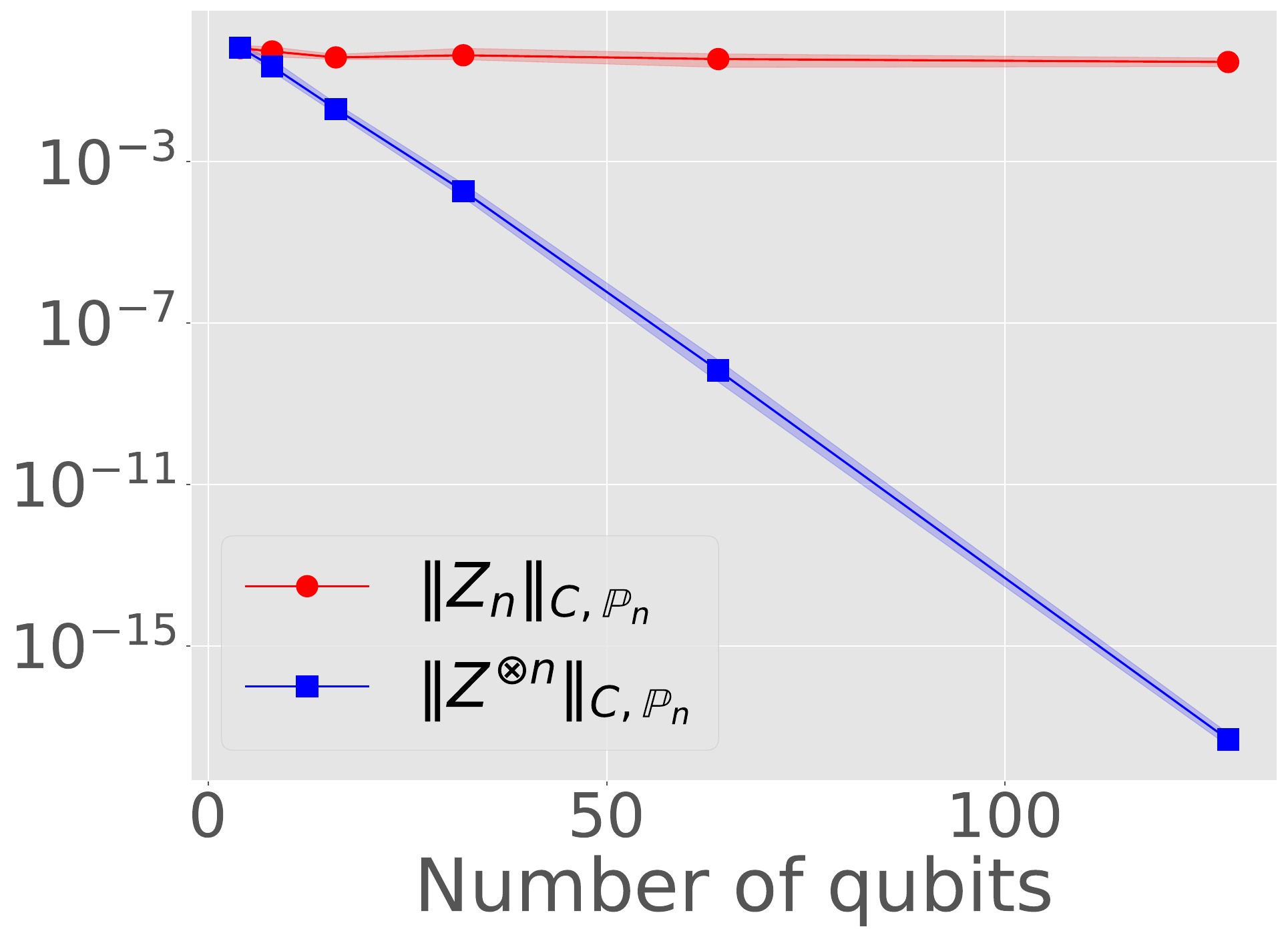}
         &
         \includegraphics[width=0.65\columnwidth]{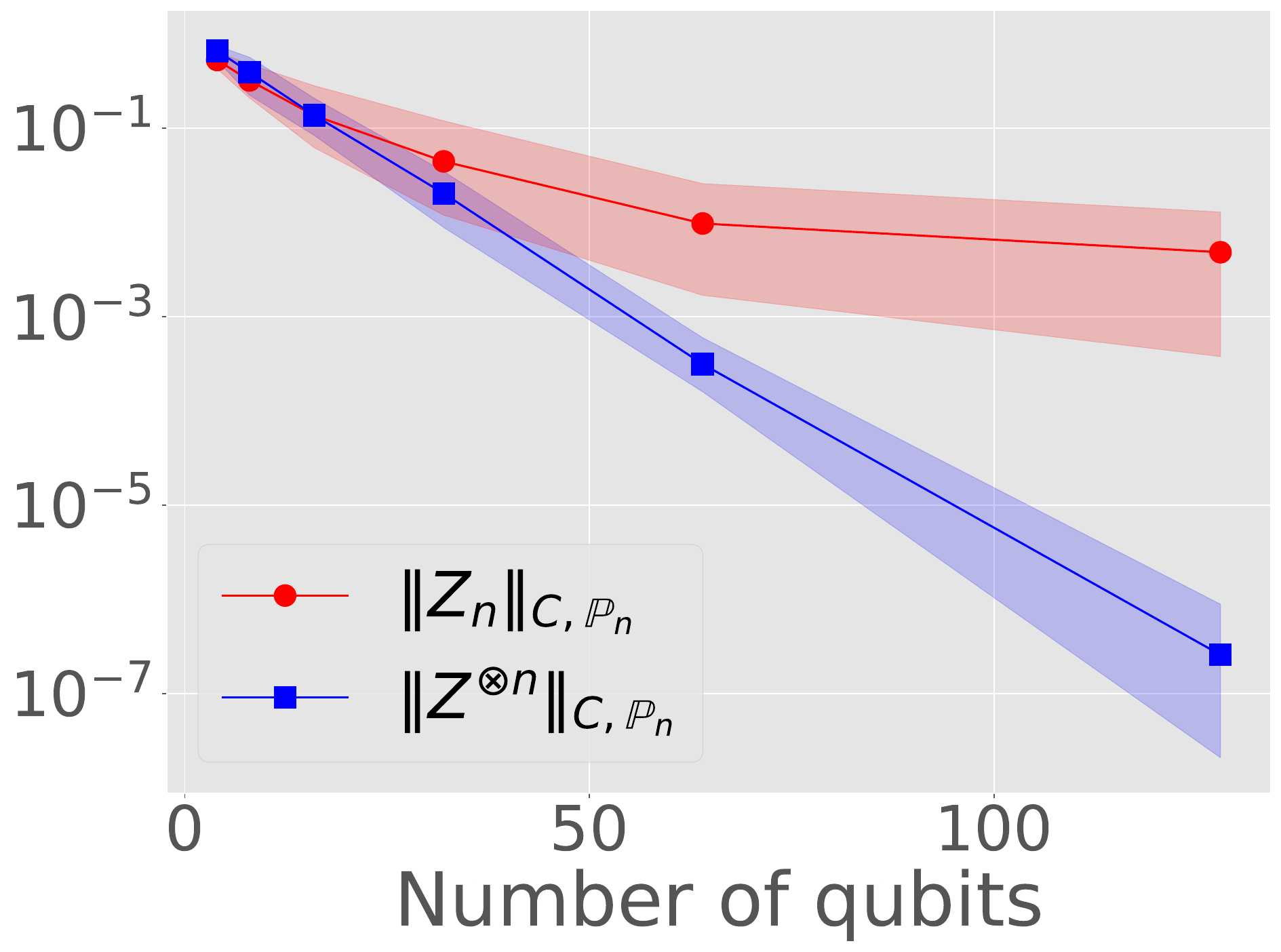}
         \\
         (d) & (e) & (f)
    \end{tabular}
    \caption{Simulation results of $C$-$\mathbb{K}$ norms and second moments. In all the plots, the x-axis represents the number of qubits. Plots (a-c) show the second moments of the cost functions while plots (d-f) show the $C$-$\mathbb{P}_n$ norms. In (a), the ansatz used is the MPS ansatz with subcircuits being HEAs with depth $\lfloor \log n \rfloor$. Here, we plot the second moments of $f_{\sigma, \ket{\mathbf{0}}\bra{\mathbf{0}}}$ and $ f_{\sigma, O}$, estimated using $10$ different $\boldsymbol{\theta}$ and averaged over five different input states randomly generated from HEAs of depth $\lfloor \log n \rfloor$. We can see that global observables induce cost concentration while local observables avoid it. In (b) and (c), we plot similar second moments for the shallow HEA and QCNN ansatzes respectively, which are following similar trends as well. In (e) and (f), we plot estimated $C$-$\mathbb{P}_n$ norms for these 2 ansatzes. From plots (b), (c), (e), and (f), we see that larger $C$-$\mathbb{P}_n$ norms are associated with trainable ansatz-observable combinations known to exhibit effective subspaces. Hence, in (d), we plot the $C$-$\mathbb{P}_n$ norms of the MPS ansatz, with subcircuit width $\lfloor \log n \rfloor $, showing a trend similar to the other ansatzes. } \label{fig:ce_norm}
    \end{figure*}

        This norm can be numerically estimated by sampling various parameter vectors $\boldsymbol{\theta}$ and taking the average of their $\| W_{C(\boldsymbol{\theta})^{\dag}}\|_{\mathbb{K}}$. Intuitively, if $\| W\|_{C,\mathbb{K}}$ remains constant or reduces only polynomially with respect to $n$, then we can expect the $C$-$W$ combination to exhibit an effect subspace since the distribution $ \mathcal{P}_{\boldsymbol{\theta}, W}$ is defined over $4^n$ outcomes. Conversely, if $ \| W\|_{C,\mathbb{K}}$ reduces exponentially with respect to $n$, then the $C$-$W$ combination need not exhibit one.

        We first numerically test this hypothesis on some instances where the presence and absence of effective subspaces are known. To do this, we choose two ansatzes; shallow HEA and QCNN, in combination with local and global observables.  
 
        It is known that effective subspaces exist when both these ansatzes are used in combination with local observables. The results (presented in Figure~\ref{fig:ce_norm}) strongly support the hypothesis and hence we carry out the same experiments for $C_T$. We discuss these simulation results in detail in the following section.

        Finally, the effective subspace for  $C_T$-$\ket{0}\bra{0}_n$ can be roughly identified by considering the cancellation of subcircuits. Typically, the probability $ \mathcal{P}_{C_T,\ket{0}\bra{0}_n,\mathbb{P}_n}(P_j)$ increases when more subcircuits are canceled within its expression, as this forces some qubits to have no circuits being acted on them and hence contribute the maximum that any qubit can to the expectation. This is also true for shallow HEA and QCNN ansatzes when used with local observables, where higher probabilities are associated with $1$-local Paulis, regardless of the position of its non-local component. For Paulis with a higher locality, one can always find an upper bound on the total number of non-canceled subcircuits that is independent of $n$ and dependent only on the locality. However, for the $C_T$-$\ket{0}\bra{0}_n$ combination, the position of the non-local part of the Pauli is crucial. The closer it is to the last qubit, the more subcircuits are canceled, resulting in higher probabilities. Similarly, if the non-local component is on the first qubit, unlike the other ansatzes, even a $1$-local observable can have no subcircuits getting canceled in the expression of the probability. Thus, the concentration of probabilities should be towards Paulis where non-local components occur near the last qubit. This hypothesis is also validated using experiments discussed in the next section. 
\begin{figure*}[tbh] 
    \centering
    \begin{tabular}{cc}
        \includegraphics[width=1\columnwidth]{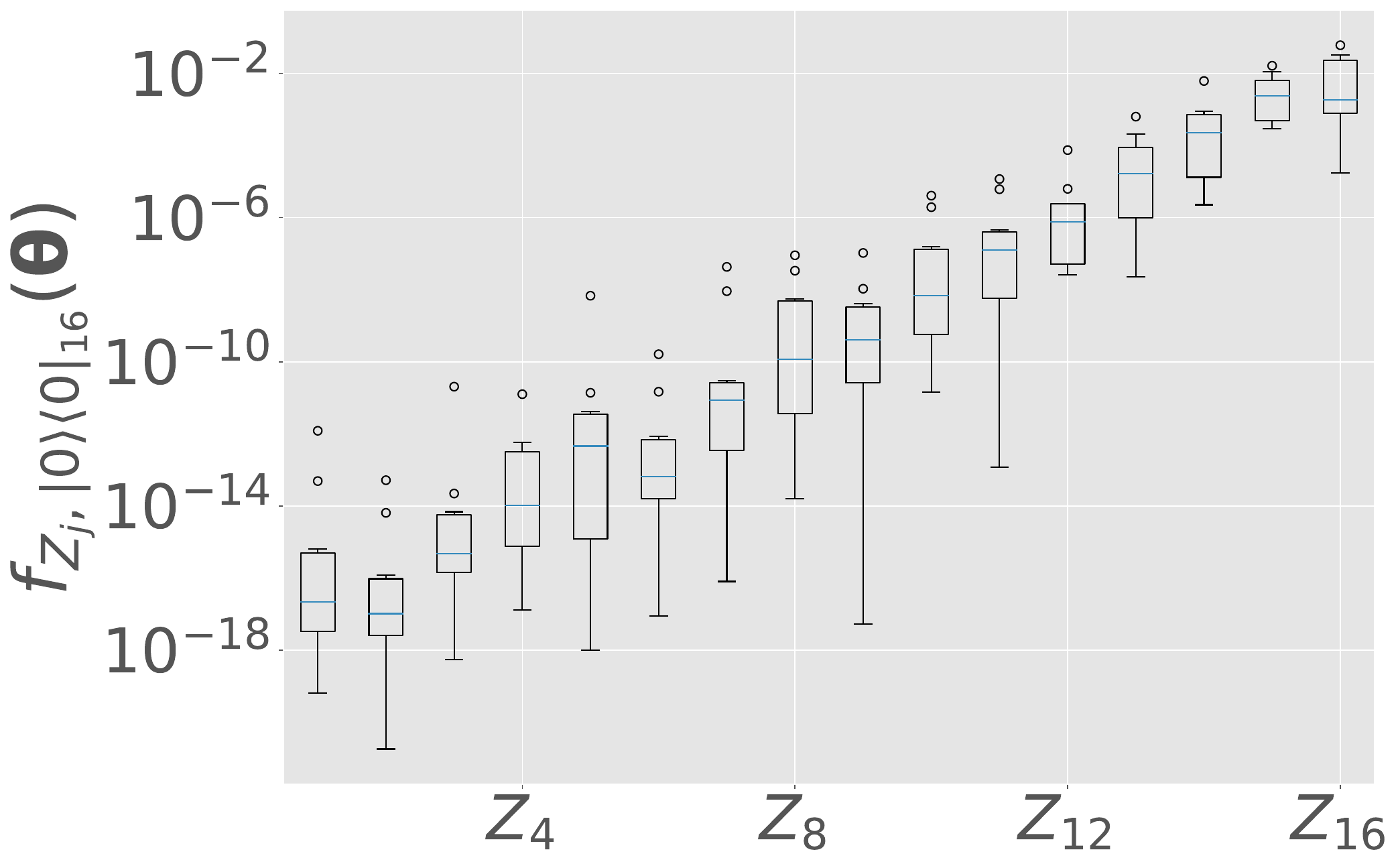} 
         &
        \includegraphics[width=1\columnwidth]{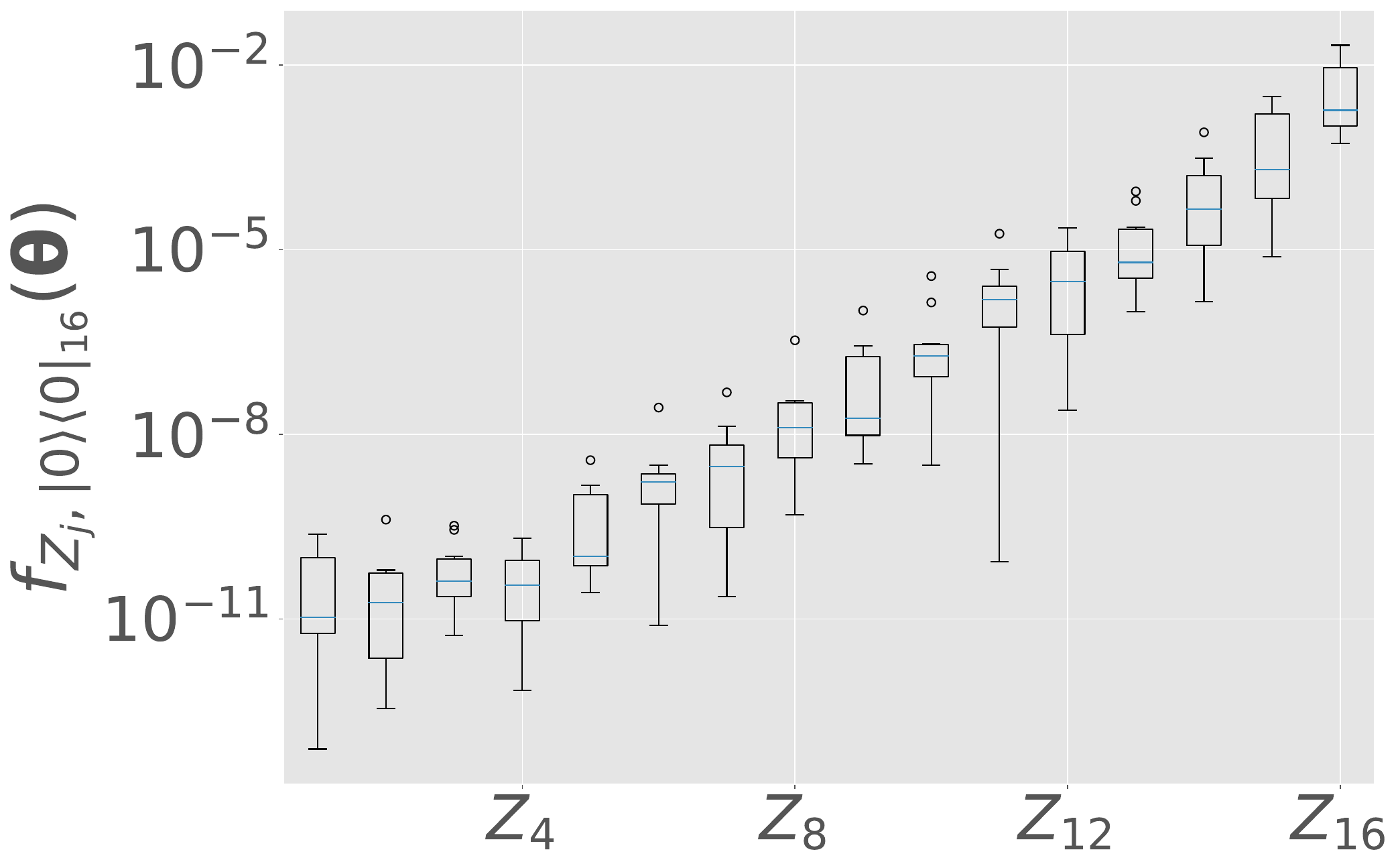} \\
        (a) $k=2$ & (b) $k=4$\\ 
        \includegraphics[width=1\columnwidth]{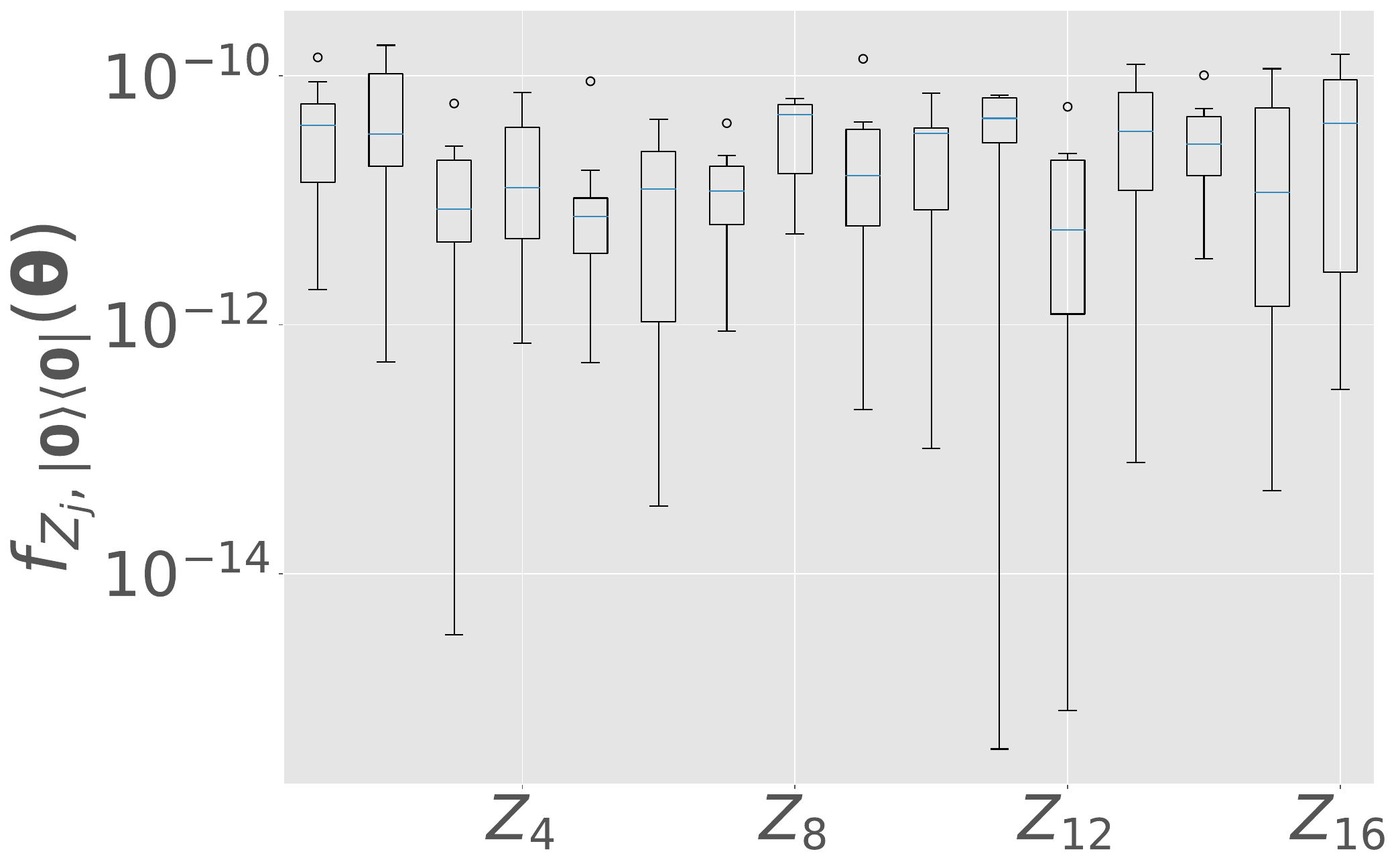} 
         &
         \includegraphics[width=1\columnwidth]{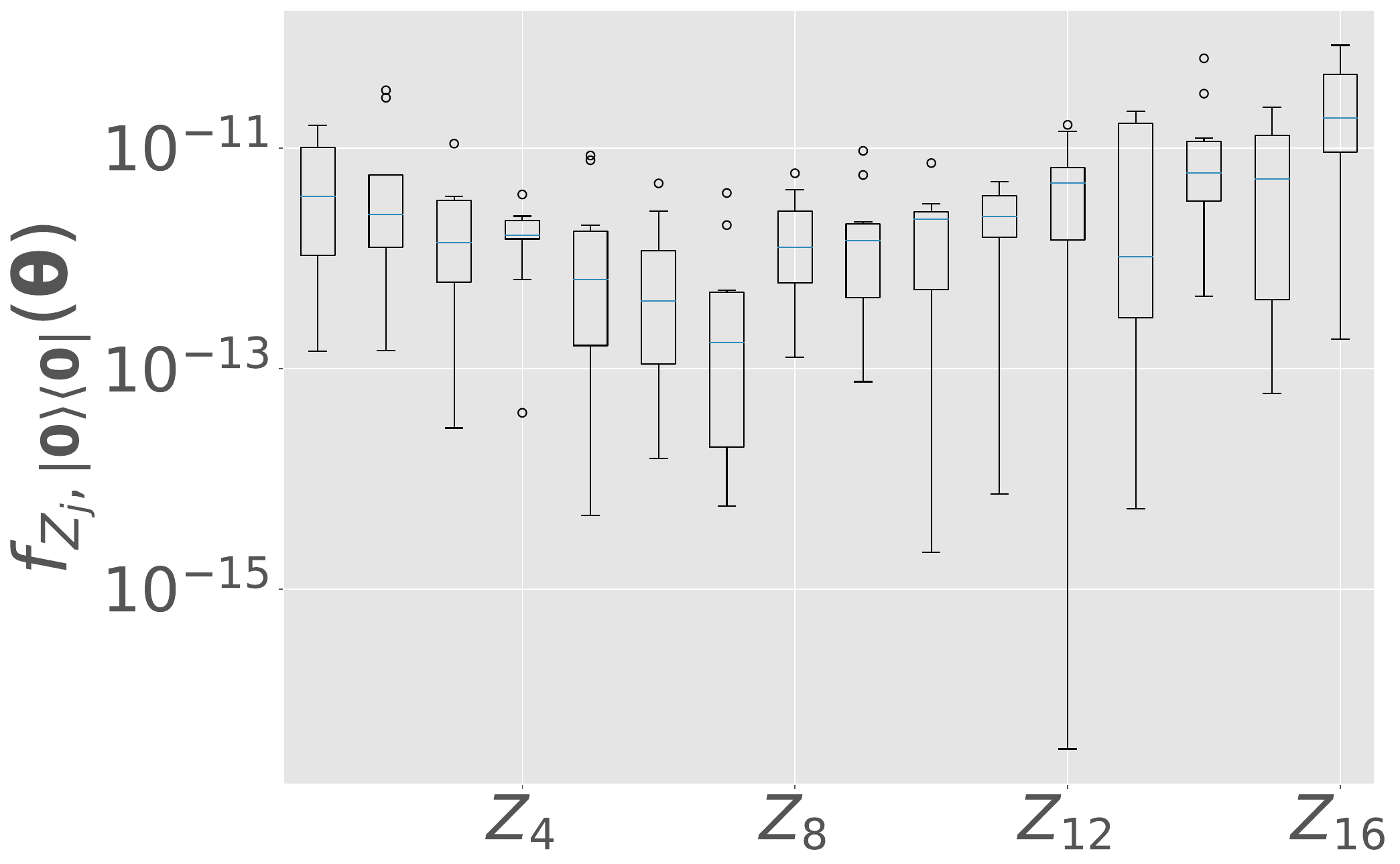} \\
        (c) $k=2$ & (d) $k=4$ 
    \end{tabular}
    \caption{ Boxplots of distributions $ \mathcal{P}_{\boldsymbol{\theta}, \ket{0}\bra{0}_{16}, \mathbb{P}_{16}}$ in (a) and (b), and distributions  $\mathcal{P}_{\boldsymbol{\theta}, \ket{\mathbf{0}} \bra{\mathbf{0}}, \mathbb{P}_{16}}$ in (c) and (d), with subcircuit widths $k=2,4$. For every $Z_i$ on the x-axis, we plot a boxplot of the probabilities computed across 10 different values of $\boldsymbol{\theta}$. The subcircuit used in each plot is an HEA with width and depth $k = 2,4$. In (a) and (b), we can see that higher probabilities are associated with $Z_i$, with $i$ close to $16$, suggesting the presence of an effective subspace consisting of these terms. Also, the distribution gets flatter as we increase $k$. In (c) and (d), we see that the distribution remains flat for all values of $k$, suggesting the absence of any effective subspace.} \label{fig:distribution}
    \end{figure*} 

\section{Simulation Results} \label{sec:simulation_results}

In this section, we discuss and present the numerical simulations that we have conducted as part of this work. The main aims of the simulations are threefold: visualize the impact of Theorems~\ref{th:upper_bound} and~\ref{th:lower_bound} using learning curves, argue that similar results could also hold for most states not necessarily satisfying the criteria mentioned in these theorems and demonstrate the presence (absence) of effective subspaces when MPS ansatzes are used with local (global) observables. The structure of all two-qubit subcircuits used in Figure~\ref{fig:ce_norm} is given in Figure~\ref{fig:ansatz} (d).

We start with the learning curves presented in Figure~\ref{fig:state_approximation}. Here, we have carried out state approximation using the MPS ansatz with subcircuit width $2$, and target state $\ket{\mathbf{0}} \bra{\mathbf{0}}$. The optimization algorithm used here is SPSA~\cite{Spall1992}, where the converging sequences are $a_j=c_j=0.4$ and all parameters are initialized uniformly from $[0,\pi/2]$. The x-axis and y-axis represent the iteration number and corresponding \textit{infidelity}, defined as $1-F$, respectively. In (a), we have plotted results for $n=20,28$, with $k=2$. We can see global observables hindering the optimization. In (b), we plot the results of similar experiments carried out for $n=12$, but with $k=2,3$ and $4$. The subcircuits are HEA with depth $k$. From this, we can see that increasing $k$ negatively impacts the optimization.

Now, we move on to Figure~\ref{fig:ce_norm} (a). Here, the x and y axes represent the number of qubits and the estimated second moments of the objective functions $f_{\sigma, \ket{\mathbf{0}}\bra{\mathbf{0}}}(\boldsymbol{\theta})$ and $ f_{\sigma, O}(\boldsymbol{\theta})$ averaged over $5$ input states randomly generated using HEA ansatz of depth $ \lfloor \log n \rfloor$, with the single qubit gates being Haar random. The subcircuit used here is HEA with width and depth $\lfloor \log n \rfloor$. We can see global observables inducing cost concentration, and local observables avoiding it, even though the input states do not necessarily satisfy the conditions required as per Theorems~\ref{th:upper_bound} and~\ref{th:lower_bound}.

Next, we move on to Figures~\ref{fig:ce_norm} (b-f). The idea here is to show that the $C$-$\mathbb{K}$ norm can be used to detect the presence of effective subspace. From~\cite{Cerezo2024}, we know that shallow HEA and QCNN ansatzes exhibit effective subspaces when used in combination with local observables. This can be seen from the plots (b), (c), (e), and (f). In (b) and (c), we have plotted the estimated second moments of the objective functions $f_{\sigma, Z^{\otimes n}}(\boldsymbol{\theta})$ and $f_{\sigma, Z_n}(\boldsymbol{\theta}) $, averaged over $5$ states generated in the same manner as in the previous experiment, for the shallow HEA and QCNN ansatzes respectively. In (e) and (f), we have plotted estimated $C$-$\mathbb{P}_n$ norms of these combinations. From these four plots, we can see that the $C$-$\mathbb{P}_n$ norms are behaving as we expected. So, in Figure~\ref{fig:ce_norm} (d), we plot the $C_T$-$ \mathbb{P}_n$ norms, with subcircuits having the same structure as in (a). The observable is chosen to be $\ket{0}\bra{0}_n$ since as mentioned earlier when it comes to classical simulation, it suffices to estimate the $C_T$-$\mathbb{P}_n$ norms of $\ket{0}\bra{0}_n$. We can see that when we use local observables, we get exponentially high $C_T$-$ \mathbb{P}_n$ norms, thus suggesting the presence of effective subspaces.

As noted at the end of Section C-$\mathbb{K}$ Norm, this effective subspace is the one that is spanned by Paulis whose non-identity components are near the last qubit. This is experimentally verified using $16$-qubit simulations whose results are shown in Figure~\ref{fig:distribution}. In (a) and (b) we plot a portion of the distribution $ \mathcal{P}_{\boldsymbol{\theta},\ket{0}\bra{0}_{16},\mathbb{P}_{16}}$  with subcircuits being HEA built using $2$ qubit Haar random gates, with depths and widths of $k=2,4$. Although there are $4^{16}$ possible outcomes, we focus on $16$, specifically the $1$-local Paulis $\{ Z_i \ | \ i = 1, \dots, n\}$, shown on the x-axis. In these figures, boxplots of probabilities $\mathcal{P}_{\boldsymbol{\theta},\ket{0}\bra{0}_{16}}(Z_i)$, computed across $10$ different $\boldsymbol{\theta}$ values are plotted. We can see that as the $Z$ component in the observables on the x-axis is closer to the last qubit, the probability is exponentially higher. In (d) and (e), similar experiments are carried out for the distribution 
$ \mathcal{P}_{\boldsymbol{\theta},\ket{\mathbf{0}}\bra{\mathbf{0}},\mathbb{P}_{16}}$, but we notice no such concentration of probabilities, indicating the absence of effective subspaces.
\section{Conclusion and Future Direction} \label{sec:conclusion}
    In this work, we have introduced new results regarding trainability and classical simulability of learning MPS approximations of quantum states variationally. We have proven that the usage of global observables forces the variance of the objective function and all its partial derivatives to be exponentially small in the number of qubits, while the usage of local observables avoids this. Moreover, we have demonstrated that using this ansatz with local observables reveals effective subspaces within the Pauli basis, paving the way for a potential classical simulation of the MPS ansatz.

    For future directions, we aim to generalize and enhance our results by theoretically proving similar trainability results when multiple layers of $C_T$ are used, extending the current proofs to all quantum states, and theoretically analyzing the effective subspaces. Additionally, we plan to develop an efficient classical simulation with rigorous performance guarantees.

\section{Code Availability}
    Code to replicate the results presented in this work can be found \href{https://github.com/afradnyf/MPS_Ansatz/tree/main}{here}.
    
\section{Acknowledgement}
    This work was partially supported by the Australian Research Council (Grant No: DP220102059). AB was partially supported by the Sydney Quantum Academy PhD scholarship. HP acknowledges the Centre for Quantum Software and Information at the University of Technology Sydney for hosting him as a visiting scholar and the support of the University of Sydney Nano Institute (Sydney Nano).
\bibliographystyle{plain}
\bibliography{references}

\section{Appendix}
We start with some notation used throughout the proofs.
\begin{itemize}
    \item For a set of Haar random unitaries $ \mathbf{U}_T = \{ U_1, \dots, U_T\}$, and an integrable function $\eta$, we define 
    \begin{align*}
        &\int \limits_{\mathbf{U}_T} \eta(\mathbf{U}_T) \text{d}\mathbf{U}_T \\ \coloneqq &\int \limits_{U_1} \dots \int \limits_{U_{T}} \eta(U_1, \dots, U_T) \text{d}U_1 \dots \text{d}U_T. \numberthis
    \end{align*}
    \item For any string $i = i_1 i_2 ,\dots i_t$, $i_{t_1:t_2} = i_{t_1} i_{t_1+1} \dots i_{t_2}$. 
    \item For any $l \in \mathbb{N}$ such that $1 \leq l \leq T-1$ and $W \in \mathbb{C}^{2^{n'} \times 2^{n'}}$ with $1 \leq n' \leq n$, define $\mu_{l}^{(W)}:\ \mathbb{C}^{2^{n'} \times 2^{n'}} \times \mathbb{C}^{2^{n'} \times 2^{n'}} \to \mathbb{C}$, where
    \begin{align}
        \mu_{l}^{(W)}(A,B) \coloneqq \int \limits_{\mathbf{U}_{l} }\text{tr}\left(W A_{C_{l}} \right) \text{tr}\left(WB_{C_{l}} \right)\text{d}\mathbf{U}_l.
    \end{align}
    \item $\mathbb{P}_n \coloneqq \{ P_{i_1 i_2\dots i_n} \coloneqq P_{i_1} \otimes \dots \otimes P_{i_n} \ | \ i_1, i_2, \dots, i_n \in \{ 0,1,2,3\}\}$, where we define $ P_0 \coloneqq \mathds1, P_1 \coloneqq X, P_2 \coloneqq Z, P_3 \coloneqq Y$ and
\end{itemize}

    Next, we introduce some definitions and lemmas that will be useful throughout our proofs.

\begin{definition} \label{def:t_design}
    A set of unitaries $\{ U(\boldsymbol{\theta}) \ | \ \boldsymbol{\theta} \in \mathbb{R}^m\}$ form a unitary $t$-design if for any polynomial (defined over matrices) $\eta(U)$ with degree at most $t$ in the matrix elements of $U$ and at most $t$ in their complex conjugates, we have
    \begin{align}
        \int \limits_{\boldsymbol{\theta}} \eta(U(\boldsymbol{\theta})) \text{d}\boldsymbol{\theta} = \int \limits_{U} \eta(U) \text{d}U,
    \end{align}
    where $\text{d}\boldsymbol{\theta} $ is the uniform distribution and $ \text{d}U$ is the Haar measure.
\end{definition}

    \begin{lemma} \label{le:eta}
        For any unitary $V$, we have 
        \begin{align}
            \int \limits_U \eta(UV) \text{d}U = \int \limits_U \eta(VU) \text{d}U = \int \limits_U \eta(U) \text{d}U,
        \end{align}
        for any integrable functional $\eta$ where $\text{d}U$ is the Haar measure. 
    \end{lemma}
    \begin{lemma} \label{le:equivalence}
        Let $k \leq n' \leq n$. For any $A,B,C \in \mathbb{C} ^ {2^{n'} \times 2^{n'}}$, we have 
        \begin{align*}
            \mu_{n'-k+1}^{(A)}(B, C)
            = \mu_{n'-k+1}^{(A)} \big( &B_{V_1 \otimes \dots \otimes V_{n'-k+1}}, \\
            &C_{V_1 \otimes \dots \otimes V_{n'-k+1}}\big) \numberthis,
        \end{align*}
    for any $k$ qubit unitary $V_1$ and $2$ qubit unitaries $V_2 \dots V_{t-k+1}$.
    \end{lemma}

    \begin{proof}
    This can be seen from the structure of $C_T^{(n)}$ in Figure~\ref{fig:ansatz} (a) and Lemma~\ref{le:eta}. Since $V_1 \otimes \dots \otimes V_{t-k+1}$ is a tensor product, the factors can be grouped with the $U_i$ gates and then can be removed by Lemma~\ref{le:eta}.
    \end{proof}
    \begin{lemma}~\cite{Mcclean2018} \label{le:1_design}
        Let $A,B\in \mathbb{C} ^ {N \times N}$ be arbitrary matrices. Then, we have
        \begin{align*}
            \int \limits_U \text{tr}(BA_U) \text{d}U = \frac{\text{tr}\left(A \right)\text{tr}\left(B \right)}{N},
        \end{align*}
        where $\text{d}U$ is the Haar measure.
    \end{lemma}
    \begin{lemma}~\cite{Mcclean2018} \label{le:2_design}
        Let $A,B,C,D \in \mathbb{C} ^ {N \times N}$ be arbitrary matrices. Then, we have
        \begin{align*}
            &\int \limits_U \text{tr}(BA_U) \text{tr}(DC_U) \text{d}U \\
            = &\frac{1}{N^2 - 1} \left( \text{tr} (A) \text{tr} (B) \text{tr} (C) \text{tr} (D) + \text{tr}(AC) \text{tr}(BD)\right) \\
            -&\frac{1}{N(N^2 - 1)} ( \text{tr} (AC) \text{tr} (B) \text{tr} (D) \\
            &+ \text{tr} (A) \text{tr} (C) \text{tr} (BD)) \numberthis,
        \end{align*}
        where $\text{d}U$ is the Haar measure.
    \end{lemma}

    \begin{lemma}~\cite{Mitarai2018} \label{le:parameter_shift_rule}
        Parameter Shift Rule: Let $\sigma \in \mathbb{D}_n$ and let $W \in \mathbb{H}_n$. Then, for any ansatz $C(\boldsymbol{\theta}) = \prod \limits_{p=1}^m  e^{-i\theta_p H_p} $, where $\boldsymbol{\theta} = [\theta_1 \dots \theta_m]^T$ and $H_p \in \mathbb{H}_n \ \forall \ p$, we have
        \begin{align}
            \partial_{\theta_p} f_{\sigma, W}(\boldsymbol{\theta}) = \frac{f_{\sigma, W}(\boldsymbol{\theta}_{p+}) - f_{\sigma, W}(\boldsymbol{\theta}_{p-})}{2},
        \end{align}
        where $\partial_{\theta_p}$ is the partial derivative with respect to $ \theta_p$, $\boldsymbol{\theta}_{p\pm} = [\boldsymbol{\theta}_{1}, \dots, \boldsymbol{\theta}_{p-1}, \boldsymbol{\theta}_{p} \pm \pi/2, \boldsymbol{\theta}_{p+1}, \dots, \boldsymbol{\theta}_{m}]^T$.
    \end{lemma}

    \begin{lemma} \cite{Cerezo2021} \label{le:exp_partial_derivative}
        Let $\sigma \in \mathbb{D}_n$ and let $W \in \mathbb{H}_n$. For any ansatz $C(\boldsymbol{\theta}) = \prod \limits_{p=1}^t U_p(\boldsymbol{\theta}_p) $ where $ U_p(\boldsymbol{\theta}_p) = \prod \limits_{q=1}^m e ^ {-i \theta_{pq} H_{pq}}$, where $\boldsymbol{\theta}_p = [\theta_1 \dots \theta_m]^T$, $H_{pq} \in \mathbb{H}_n$ and $ \boldsymbol{\theta} = \boldsymbol{\theta}_1 \oplus \dots \oplus \boldsymbol{\theta}_t$, and for any $p,q$, define
    \begin{align*}
        U_p^{(L,q)} \left(\boldsymbol{\theta}_p \right) &= \prod \limits_{j=1}^{q-1} e^{-i \theta_{pj} H_{pj}}, \\
        U_p^{(R,q)} \left(\boldsymbol{\theta}_p \right) &= \prod \limits_{j=q+1}^{m} e^{-i \theta_{pj} H_{pj}}. \numberthis
    \end{align*}
    Then, we have
        \begin{align}
            \mathbb{E}_{\boldsymbol{\theta}} \left( \partial_{\theta_{pq}} f_{\sigma, W} (\boldsymbol{\theta}) \right) = 0,
        \end{align}
        where $ U_1({\boldsymbol{\theta}_1}), U_2({\boldsymbol{\theta}_2}), \dots, U_{p-1}({\boldsymbol{\theta}_{p-1}}), U_{p+1} ({\boldsymbol{\theta}_{p+1}}), \dots, $ $ U_{t}({\boldsymbol{\theta}_{t}})$ along with either $ U_p^{(L,q)}$ or $ U_p^{(R,q)}$ form unitary $2$-designs.
    \end{lemma}

    \begin{lemma} \label{le:matrix_holder}
        Tracial Matrix H\"{o}lder's Inequality~\cite{Baumgartner2011}: For any $A,B \in \mathbb{C}^{t \times t}$, and any $1\leq p,q\leq \infty$ such that $\frac{1}{p} + \frac{1}{q} = 1$, we have,
        \begin{align}
            \left| \text{tr}(A^{\dag}B) \right| \leq \| A\|_p \| B\|_q.
        \end{align}
    \end{lemma}
    \begin{lemma} \label{le:cross_product}
        For a set of numbers $\lambda_0, \lambda_1, \dots, \lambda_{D-1}$, where $D$ is even, with $\lambda_i \in \{ 1, -1\}$ and $|\{ \lambda_i = 1 \ | \ i = 0, \dots, D-1\}| = D/2$, we have $\sum \limits_{i,j=0, i \neq j}^{D-1} \lambda_i \lambda_j = -D$.
    \end{lemma}

    \begin{proof}
        Let $\mathbb{I} = \{ (i,j) \ | \ i,j = 0, 1, 2, \dots, D-1\}$. Define $\mathbb{I}_{\pm} = \{ (i,j) \in \mathbb{I} \ | \lambda_i \lambda_j = \pm 1\}$. So $| \mathbb{I}_{+}| = | \mathbb{I}_{-}| = D^2/2 $. All $\frac{D^2}{2}$ elements in $\mathbb{I}_- $ should have $i\neq j$. But only $\frac{D^2}{2} - D$ elements in $\mathbb{I}_+$ can have $i \neq j$ since whenever $i=j$, $(i,j) \in \mathbb{I}_+$. Hence, $\sum \limits_{i \neq j} \lambda_{i} \lambda_{j} = -D$. 
    \end{proof}

    \begin{lemma} \label{le:1_design_multi}
        Let $A \in \mathbb{C}^{2^n \times 2^n}$. Then, we have,
        \begin{align}
            \int \limits_{U} \text{tr}\left(Z_i A_{\mathds1_{n-k} \otimes U} \right) \text{d}U = 0,
        \end{align}
        for any $i \in \{ n-k+1,\dots, n\}$, where $\text{d}U$ is the Haar measure.
    \end{lemma}
    \begin{proof}
        Let $A = \sum \limits_{p,q=0}^{2^n-1} A_{pq} \ket{p_{1:n-k}}\bra{q_{1:n-k}} \otimes $ \\ $ \ket{p_{n-k+1:n}}\bra{q_{n-k+1:n}}$. Then, we have
        \begin{align*}
            &\int \limits_{U} \text{tr} \left(Z_i A_{\mathds1_{n-k} \otimes U} \right) \text{d}U \\
            = &\sum \limits_{p,q=0}^{2^n-1} A_{pq} \delta_{p_{1:n-k},q_{1:n-k}} \\
            &\hspace*{1cm}\int \limits_{U} \text{tr} \big(Z_i (\ket{p_{n-k+1:n}} \bra{q_{n-k+1:n}})_{U} \big)\text{d}U \\
            = &0 \numberthis,
        \end{align*}
        where the last equality follows from Lemma~\ref{le:1_design}.
    \end{proof}
    \begin{lemma} \label{le:2_design_multi}
        Let $A \in \mathbb{C}^{2^n \times 2^n}$. Then, we have,
        \begin{align}
            \int \limits_{U} \text{tr}(Z_i A_{\mathds1_{n-k} \otimes U}) \text{tr}(Z_j A_{\mathds1_{n-k} \otimes U}) \text{d}U = 0,
        \end{align}
        for any $i,j \in \{ n-k+1,\dots, n\}$ with $i \neq j$, where $\text{d}U$ is the Haar measure.
    \end{lemma}
    \begin{proof}
        Let $A = \sum \limits_{p,q=0}^{2^n-1} A_{pq} \ket{p_{1:n-k}}\bra{q_{1:n-k}} \otimes $ \\ $ \ket{p_{n-k+1:n}}\bra{q_{n-k+1:n}}$. Then, we have
        \begin{align*}
            &\int \limits_{U} \text{tr} \left( Z_i A_{\mathds1_{n-k} \otimes U} \right) \text{tr} \left( Z_jA_{\mathds1_{n-k} \otimes U} \right) \text{d}U \\
            = &\int \limits_{U} \sum \limits_{p,q,r,s=0}^{2^n-1} A_{pq} A_{rs} \delta_{p_{1:n-k},q_{1:n-k}} \delta_{r_{1:n-k},s_{1:n-k}} \\
            &\text{tr}(Z_i(\ket{p_{n-k+1:n}}\bra{q_{n-k+1:n}})_{U}) \\&\text{tr}(Z_j(\ket{r_{n-k+1:n}}\bra{s_{n-k+1:n}})_{U}) \text{d}U = 0, \numberthis
        \end{align*}
        where the last equality follows from Lemma~\ref{le:2_design}.
    \end{proof}
\begin{customthm}{1}
    Let $\sigma \in \mathbb{D}_n$ and $C_T^{(n)}$ be an MPS ansatz where each parameterized subcircuit $U_i$ forms a unitary $2$-design. Then, we have
    \begin{align}
        \text{Var}_{\boldsymbol{\theta}} \left(  f_{\sigma,\ket{\mathbf{0}}\bra{\mathbf{0}}}(\boldsymbol{\theta}) \right) \leq \frac{h_1(\sigma)}{4^ {n-k-1}}.
    \end{align}
\end{customthm}

\begin{proof}
    Note that since 
    \begin{align}
        \text{Var}_{\boldsymbol{\theta}} \left( f_{\sigma,\ket{\mathbf{0}}\bra{\mathbf{0}}}(\boldsymbol{\theta}) \right) \leq \mathbb{E}_{\boldsymbol{\theta}} \left(f_{\sigma,\ket{\mathbf{0}}\bra{\mathbf{0}}}(\boldsymbol{\theta}) \right)^2,
    \end{align}
    it is sufficient to prove $\mu_{T}^{\left(\ket{\mathbf{0}} \bra{\mathbf{0}} \right)}(\sigma, \sigma) \leq \left( h_1(\sigma)/4 ^ {n-k-1} \right)$. 

    Let $\sigma = \sum \limits_{ij} \sigma_{ij} \ket{i} \bra{j}$. Then, we have
    \begin{align*} \label{eq:sigma_sigma}
        &\mu_{T}^{\left(\ket{\mathbf{0}} \bra{\mathbf{0}}\right)}(\sigma, \sigma) \\\
        = &\sum \limits_{pqrs} \sigma_{pq} \sigma_{rs} \mu_{T}^{\left(\ket{\mathbf{0}} \bra{\mathbf{0}}\right)}(\ket{p}\bra{q},\ket{r}\bra{s}) \\
        \leq &\sum \limits_{pqrs} | \sigma_{pq} ||\sigma_{rs} | |\mu_{T}^{\left(\ket{\mathbf{0}} \bra{\mathbf{0}}\right)}(\ket{p}\bra{q},\ket{r}\bra{s})| \numberthis.
    \end{align*} 
    
    Our next step is to prove that $ |\mu_{T}^{\left(\ket{\mathbf{0}} \bra{\mathbf{0}}\right)}(\ket{p}\bra{q},\ket{r}\bra{s})| $ is upper bounded by $1/4 ^ {n-k-1}$.  The main idea is to integrate each unitary  $U_T, U_{T-1}, \dots, U_1$ one at a time, in this order. For this, we use a recursive method, where we break down this integral into an integration over $U_T$ multiplied by integrations over an MPS circuit defined with $T-1$ Haar random $k$-qubit gates. We then show that such an integration can be solved efficiently for any $t$, thus allowing us to carry out the full integration efficiently.
    
    Let $p = p_1 p_2 \dots p_n$ be the binary expansion of $p$ (similar definitions for $q,r$ and $s$). Then, we have
    \begin{align*} \label{eq:upper_bound_pqrs}
        &\mu_{T}^{\left(\ket{\mathbf{0}} \bra{\mathbf{0}}\right)}(\ket{p} \bra{q}, \ket{r}\bra{s}) \\ 
        =&\int \limits_{\mathbf{U}_T} \bra{\mathbf{0}} \left( \ket{p}\bra{q} \right)_{C_T} \ket{\mathbf{0}} \bra{\mathbf{0}} \left( \ket{r}\bra{s} \right) _{C_T} \ket{\mathbf{0}} \text{d}\mathbf{U}_T \numberthis \\
        = &\int \limits_{\mathbf{U}_T} \bra{\mathbf{0}} \left( \ket{p_{1:k}} \bra{q_{1:k}}_{U_T} \otimes \ket{p_{k+1:n}} \bra{q_{k+1:n}} \right)_{C_{T-1}} \ket{\mathbf{0}} \\
        &\bra{\mathbf{0}} \big( \ket{r_{1:k}} \bra{s_{1:k}}_{U_T} \otimes \ket{r_{k+1:n}} \bra{s_{k+1:n}} \big)_{C_{T-1}} \ket{\mathbf{0}} \\
        &\text{d}\mathbf{U}_T \numberthis \\
        \end{align*}
        Expanding the above in the Pauli basis gives us
        \begin{align*}
        &\mu_{T}^{\left(\ket{\mathbf{0}} \bra{\mathbf{0}}\right)}(\ket{p} \bra{q}, \ket{r}\bra{s}) \\
        = &\frac{1}{2^{2k}} \int \limits_{\mathbf{U}_T} \sum \limits_{\substack{i_{1:k} \\ j_{1:k}}} \bra{\mathbf{0}}\left( P_{i_{1:k}} \otimes \ket{p_{k+1:n}} \bra{q_{k+1:n}} \right)_{C_{T-1}} \ket{\mathbf{0}} \\
        &\text{tr} \left(P_{i_{1:k}} \ket{p_{1:k}} \bra{q_{1:k}}_{U_T}\right)\bra{\mathbf{0}} ( P_{j_{1:k}} \otimes \ket{r_{k+1:n}} \\ &\bra{s_{k+1:n}} )_{C_{T-1}} \ket{\mathbf{0}}\text{tr} \left(P_{j_{1:k}}\ket{r_{1:k}} \bra{s_{1:k}}_{U_T} \right) \text{d}\mathbf{U}_T \numberthis \\
        \end{align*}
        Now we isolate the integral over $U_T$ and set up the recursion.
        \begin{align*}
        &\mu_{T}^{\left(\ket{\mathbf{0}} \bra{\mathbf{0}}\right)}(\ket{p} \bra{q}, \ket{r}\bra{s}) \\
        = &\frac{1}{2^k} \int \limits_{\mathbf{U}_{T-1}} \sum \limits_{\substack{i_{1:k} \\ j_{1:k}}} \bra{\mathbf{0}} \left( P_{i_{1:k}} \otimes \ket{p_{k+1:n}} \bra{q_{k+1:n}}\right)_{C_{T-1}} \ket{\mathbf{0}} \\
        &\hspace*{1.8cm}\bra{\mathbf{0}} \left( P_{j_{1:k}} \otimes \ket{r_{k+1:n}} \bra{s_{k+1:n}} \right)_{C_{T-1}} \ket{\mathbf{0}} \\
        & \hspace*{1.8cm}\frac{1}{2^k} \int \limits_{U_T} \text{tr} \left( P_{i_{1:k}}\ket{p_{1:k}} \bra{q_{1:k}}_{U_T} \right) \\
        &\hspace*{2.8cm}\text{tr} \left(P_{j_{1:k}} \ket{r_{1:k}} \bra{s_{1:k}}_{U_T}\right) \text{d}\mathbf{U}_T \numberthis \\
        = &\frac{1}{2^k} \sum \limits_{\substack{i_{1:k} \\ j_{1:k}}} \mu_{T-1}^{(\ket{\mathbf{0}} \bra{\mathbf{0}})}(P_{i_{1:k}} \otimes \ket{p_{k+1:n}} \bra{q_{k+1:n}}, \\
        &\hspace*{2.3cm}P_{j_{1:k}} \otimes \ket{r_{k+1:n}} \bra{s_{k+1:n}}) \\
        &\hspace*{1.1cm}\frac{1}{2^k} \int \limits_{U_T} \text{tr} \left( P_{i_{1:k}} \ket{p_{1:k}} \bra{q_{1:k}}_{U_T} \right) \\ 
        &\hspace*{2cm}\text{tr} \left( P_{j_{1:k}} \ket{r_{1:k}} \bra{s_{1:k}}_{U_T} \right) \text{d}\mathbf{U}_T. \numberthis
    \end{align*}

    Now that we have isolated the integration over $U_T$, we use Lemma~\ref{le:2_design} to integrate over $U_T$ and get expressions that look like $\mu_{T}^{\left(\ket{\mathbf{0}} \bra{\mathbf{0}}\right)}(\ket{p} \bra{q}, \ket{r}\bra{s})$, but with different bit strings, defined over an $n-1$ qubit system, and as an integral of $T-1$ Haar random gates. 
    
    Notice that when 
    $b_1 \notin \{0,2\}$, $ \bra{\mathbf{0}} \left(P_{b_{1:k}} \otimes \ket{p_{k+1:n}} \bra{q_{k+1:n}} \right)_{C_{T-1}} \ket{\mathbf{0}}=0 $. Similarly, using Lemma~\ref{le:2_design}, we can see that when $P_{i_{1:k}}$ and $P_{j_{1:k}}$ are not equal, $\int \limits_{U_T} \text{tr} \left( P_{i_{1:k}} \ket{p_{1:k}} \bra{q_{1:k}}_{U_T} \right) \text{tr} \left( P_{j_{1:k}} \ket{r_{1:k}} \bra{s_{1:k}}_{U_T} \right) = 0$. 

    Define 
    \begin{align}
        \delta_{\text{ip}}^{(1)} &= \delta_{p_{1:k},r_{1:k}} \delta_{q_{1:k},s_{1:k}} \\ \delta_{\text{tr}}^{(1)} &= \delta_{p_{1:k},q_{1:k}} \delta_{r_{1:k},s_{1:k}}.
    \end{align}
    When $ i_{1:k} = j_{1:k} = 0\dots 0$, we have
    \begin{align*}
        \frac{1}{2^k} \int \limits_{U_T} &\text{tr} \left( P_{i_{1:k}} \ket{p_{1:k}} \bra{q_{1:k}}_{U_T} \right) \\
        &\text{tr} \left( P_{j_{1:k}} \ket{r_{1:k}} \bra{s_{1:k}}_{U_T} \right) \text{d}U_T = \frac{\delta_{\text{tr}}^{(1)}}{2^k} \numberthis
    \end{align*}
    and when $ i_{1:k} = j_{1:k} \neq 0\dots 0$ with $ i_1 \in \{ 0,2\}$, we have
    \begin{align*}
        &\frac{1}{2^k} \int \limits_{U_T} \text{tr} \left( P_{i_{1:k}} \ket{p_{1:k}} \bra{q_{1:k}}_{U_T} \right) \\
        &\hspace*{0.9cm}\text{tr} \left( P_{j_{1:k}} \ket{r_{1:k}} \bra{s_{1:k}}_{U_T} \right) \text{d}U_T \numberthis \\
        = &\frac{1}{2^{2k}-1}\delta_{\text{ip}}^{(1)}  - \frac{1}{2^k(2^{2k}-1)}  \delta_{\text{tr}}^{(1)} \numberthis \\
        = &\frac{2^k\delta_{\text{ip}}^{(1)} - \delta_{\text{tr}}^{(1)}}{2^k(2^{2k}-1)} \numberthis \\
        = &\tau. \numberthis
    \end{align*}
     Hence, we have
    \begin{align*}
    &\mu_{T}^{\left(\ket{\mathbf{0}} \bra{\mathbf{0}}\right)}(\ket{p} \bra{q}, \ket{r}\bra{s}) \\
        =&\frac{\delta_{\text{tr}}^{(1)}}{2^{2k}} \mu_{T-1}^{(\ket{0}\bra{0} ^ {\otimes n-1})}(\mathds1_{k-1} \otimes \ket{p_{k+1:n}} \\
        &\bra{q_{k+1:n}}, \mathds1_{k-1} \otimes \ket{r_{k+1:n}} \bra{s_{k+1:n}})
        \\
        &+ \frac{\tau}{2^k}\sum \limits_{\substack{P \in \mathbb{P}_{k-1} \\ P \neq \mathds1_{k-1}}}\mu_{T-1}^{(\ket{0}\bra{0} ^ {\otimes n-1})}(P \otimes \ket{p_{k+1:n}} \\
        &\bra{q_{k+1:n}}, P \otimes \ket{r_{k+1:n}} \bra{s_{k+1:n}}). \numberthis
    \end{align*}
    Now, using Lemma~\ref{le:equivalence}, we have
\begin{align*}
        &\mu_{T}^{\left(\ket{\mathbf{0}} \bra{\mathbf{0}}\right)}(\ket{p} \bra{q}, \ket{r}\bra{s}) \\
        =&\frac{\delta_{\text{tr}}^{(1)}}{2^{2k}} \mu_{T-1}^{(\ket{0}\bra{0} ^ {\otimes n-1})}(\mathds1_{k-1} \otimes \ket{p_{k+1:n}} \bra{q_{k+1:n}}, \\
        &\hspace*{2.6cm}\mathds1_{k-1} \otimes \ket{r_{k+1:n}} \bra{s_{k+1:n}}) \\
        &+ \frac{(2^{2k-2}-1)\tau }{2^k}\mu_{T-1}^{(\ket{0}\bra{0} ^ {\otimes n-1})}(Z^{\otimes k-1} \otimes \ket{p_{k+1:n}} \\
        &\hspace*{0.4cm}\bra{q_{k+1:n}}, Z^{\otimes k-1} \otimes \ket{r_{k+1:n}} \bra{s_{k+1:n}}). \numberthis
    \end{align*}
    This is because there always exists a unitary that will map any non-identity Pauli to any other non-identity Pauli.
    Let $\mathds1_{k-1} = \sum \limits_{i=0}^{2^{k-1}-1} \ket{i} \bra{i}$ and $ Z^{\otimes (k-1)} = \sum \limits_{i}^{2^{k-1}-1} \lambda_{i} \ket{i} \bra{i}$ be spectral decompositions. Then, we have
    \begin{align*}
    &\mu_{T}^{\left(\ket{\mathbf{0}} \bra{\mathbf{0}}\right)}(\ket{p} \bra{q}, \ket{r}\bra{s}) \\
        =&\frac{\delta_{\text{tr}}^{(1)}}{2^{2k}} \sum \limits_{i,j=0}^{2^{k-1}-1} \mu_{T-1}^{\left(\ket{0}\bra{0} ^ {\otimes n-1}\right)}(\ket{i} \bra{i} \otimes \ket{p_{k+1:n}} \bra{q_{k+1:n}}, \\
        &\hspace*{3.6cm}\ket{j} \bra{j} \otimes \ket{r_{k+1:n}} \bra{s_{k+1:n}}) \\
        &+ \frac{(2^{2k-2} - 1)\tau}{2^k} \sum \limits_{i,j=0}^{2^{k-1}-1} \lambda_i \lambda_j \mu_{T-1}^{(\mathbf{0} ^ {\otimes n-1})} (\ket{i} \bra{i} \otimes \ket{p_{k+1:n}} \\
        &\hspace*{0.4cm}\bra{q_{k+1:n}}, \ket{j} \bra{j} \otimes \ket{r_{k+1:n}} \bra{s_{k+1:n}}). \numberthis
    \end{align*}

    Now, using Lemma~\ref{le:equivalence}, we also have that
    \begin{align*}
        &\mu_{T-1}^{(\ket{0}\bra{0} ^ {\otimes n-1})}(\ket{i} \bra{i} \otimes \ket{p_{k+1:n}} \bra{q_{k+1:n}}, \\
        &\hspace*{1.9cm}\ket{i} \bra{i} \otimes \ket{r_{k+1:n}} \bra{s_{k+1:n}}) \\
        =&\mu_{T-1}^{(\ket{0}\bra{0} ^ {\otimes n-1})}(\ket{0}\bra{0} ^ {\otimes k-1}\otimes \ket{p_{k+1:n}} \bra{q_{k+1:n}}, \\
        &\hspace*{1.9cm}\ket{0}\bra{0} ^ {\otimes k - 1} \otimes \ket{r_{k+1:n}} \bra{s_{k+1:n}}) \numberthis
    \end{align*}
    and
    \begin{align*}
        & \mu_{T-1}^{(\ket{0}\bra{0} ^ {\otimes n-1})}(\ket{i} \bra{i} \otimes \ket{p_{k+1:n}} \bra{q_{k+1:n}}, \\
        &\hspace*{1.9cm}\ket{j} \bra{j} \otimes \ket{r_{k+1:n}} \bra{s_{k+1:n}}) \\
        = &\mu_{T-1}^{(\ket{0}\bra{0} ^ {\otimes n-1})}(\ket{0}\bra{0} ^ {\otimes k-1} \otimes \ket{p_{k+1:n}} \bra{q_{k+1:n}}, \\
        &\hspace*{1.9cm}\ket{1} \bra{1} ^ {\otimes k-1} \otimes \ket{r_{k+1:n}} \bra{s_{k+1:n}}). \numberthis
    \end{align*}

    The reason for the second equation is that we can always find a $k-1$-qubit unitary $V$ such that $V\ket{i} \bra{i} V^{\dag} = \ket{0} \bra{0} ^ {\otimes k-1}$ and $V\ket{j} \bra{j} V^{\dag} = \ket{1} \bra{1}^ {\otimes k-1}$. A similar explanation for the first equation as well.

    Now, define ${\Delta_t}^{(=)}$ and ${\Delta_t}^{(\neq)}$ as
    \begin{align*}
        {\Delta_t}^{(=)} &= \mu_{T-t}^{(\ket{0}\bra{0} ^ {\otimes n-t})}(\ket{0}\bra{0} ^ {\otimes k-1} \otimes \ket{p_{k+t:n}} \\ 
        &\hspace*{0.4cm}\bra{q_{k+t:n}},\ket{0}\bra{0} ^ {\otimes k-1} \otimes \ket{r_{k+t:n}} \bra{s_{k+t:n}}) \numberthis \\
        {\Delta_t}^{(\neq)} &=  \mu_{T-t}^{(\ket{0}\bra{0} ^ {\otimes n-t})}(\ket{0}\bra{0} ^ {\otimes k-1} \otimes \ket{p_{k+t:n}} \\ 
        &\hspace*{0.4cm}\bra{q_{k+t:n}},\ket{1}\bra{1} ^ {\otimes k-1} \otimes \ket{r_{k+t:n}} \bra{s_{k+t:n}}). \numberthis
    \end{align*}

    Then, we have
    \begin{align*}
        &\mu_{T-1}^{(\ket{0}\bra{0} ^ {\otimes n-1})}(\mathds1_{k-1} \otimes \ket{p_{k+1:n}} \bra{q_{k+1:n}}, \\
        &\hspace*{1.9cm}\mathds1_{k-1} \otimes \ket{r_{k+1:n}} \bra{s_{k+1:n}}) \\
        = &2^{k-1} \Delta_1^{(=)} + 2^{k-1} (2^{k-1} - 1) \Delta_1^{(\neq)}. \numberthis
    \end{align*}

    Similarly, we have
    \begin{align*}
        &\mu_{T-1}^{(\ket{0}\bra{0} ^ {\otimes n-1})}(Z^{\otimes k-1} \otimes \ket{p_{k+1:n}} \bra{q_{k+1:n}}, \\
        &\hspace*{1.9cm}Z^{\otimes k-1} \otimes \ket{r_{k+1:n}} \bra{s_{k+1:n}}) \\
        = &2^{k-1} \Delta_1^{(=)} + \sum \limits_{i,j=0, i \neq j} ^ {2^{k-1}-1} \lambda_i \lambda_j \Delta_1^{(\neq)} \\
        = &2^{k-1} \Delta_1^{(=)} -2^{k-1} \Delta_1^{(\neq)}. \numberthis
    \end{align*}
    The reason for the last equality is as follows. Notice that the set $\{ \lambda_i \ | \ i=0, \dots, 2^{k-1}\}$ has $2^{k-2}$ $1$s and $2^{k-2}$ $-1$s. So, using Lemma~\ref{le:cross_product} we see that $\sum \limits_{i \neq j} \lambda_{i} \lambda_{j} = -2^{k-1}$. 

    Hence, we have
    \begin{align}
        \mu_{T}^{\left(\ket{\mathbf{0}} \bra{\mathbf{0}}\right)}(\ket{p} \bra{q}, \ket{r}\bra{s}) = \alpha_1\Delta_1^{(=)} +  \beta_1\Delta_1^{(\neq)},
    \end{align}
    where
    \begin{align*}
        \alpha_1 = &\frac{\delta_{\text{tr}}^{(1)}}{2^{k+1}} + \frac{(2^{2k-2}-1)(2^k\delta_{\text{ip}}^{(1)} - \delta_{\text{tr}}^{(1)})}{2^{k+1}(2^{2k}-1)}, \numberthis \\
        \beta_1 = &\frac{\delta_{\text{tr} }^{(1)}(2^{k-1}-1)}{2^{k+1}} \\
        &- \frac{(2^{2k-2}-1)(2^k\delta_{\text{ip}}^{(1)} - \delta_{\text{tr}}^{(1)})}{2^{k+1}(2^{2k}-1)} \numberthis.
    \end{align*}
    So, we have achieved the goal of reducing the Eq~\eqref{eq:upper_bound_pqrs} to a similar integral of $n-1$ qubit systems and $T-1$ Haar random gates. 

    Now, we can integrate $\Delta_1^{(=)}$ and $\Delta_1^{(\neq)}$ in the same way. For that, we first define
    \begin{align*}
        \alpha_t^{(=)} &= \frac{\delta_{\text{tr}}^{(1)}}{2^{k+1}} + \frac{(2^{2k-2}-1)(2^k\delta_{\text{ip}}^{(t)} - \delta_{\text{tr}}^{(t)})}{2^{k+1}(2^{2k}-1)}, \numberthis \\
        \beta_t^{(=)} &= \frac{\delta_{\text{tr} }^{(1)}(2^{k-1}-1)}{2^{k+1}} \\
        &- \frac{(2^{2k-2}-1)(2^k\delta_{\text{ip}}^{(t)} - \delta_{\text{tr}}^{(t)})}{2^{k+1}(2^{2k}-1)} \numberthis \\
        \alpha_t^{(\neq)} &= \frac{\delta_{\text{tr}}^{(1)}}{2^{k+1}} + \frac{(2^{2k-2}-1)( - \delta_{\text{tr}}^{(t)})}{2^{k+1}(2^{2k}-1)} \numberthis \\
        \beta_t^{(\neq)} &= \frac{\delta_{\text{tr} }^{(1)}(2^{k-1}-1)}{2^{k+1}} - \frac{(2^{2k-2}-1)( - \delta_{\text{tr}}^{(t)})}{2^{k+1}(2^{2k}-1)},\numberthis
    \end{align*}
    for $2 \leq t \leq T$ and 
    \begin{align}
        \delta_{\text{ip}}^{(t)} = \delta_{p_{k+t-1},r_{k+t-1}} \delta_{q_{k+t-1},s_{k+t-1}} \\
        \delta_{\text{tr}}^{(t)} = \delta_{p_{k+t-1},q_{k+t-1}} \delta_{r_{k+t-1},s_{k+t-1}}.
    \end{align}
    
    Integrating $\Delta_1^{(=)}$ will result in $\alpha_2^{(=)}\Delta_2^{(=)} +  \beta_2^{(=)}\Delta_2^{(\neq)}$ and integrating $\Delta_1^{(\neq)}$ will result in $\alpha_2^{(\neq)}\Delta_2^{(=)} + \beta_2^{(\neq)}\Delta_2^{(\neq)}$.
    
    Assume that after integration over unitaries $U_T, \dots U_{T-t+1}$, we get $ \gamma^{(=)} \Delta_t^{(=)} + \gamma^{(\neq)} \Delta_t^{(\neq)}$. Now, if we integrate over $U_{T-t}$, the coefficients of $\Delta_{t+1}^{(=)}$ and $\Delta_{t+1}^{(\neq)}$ will be $ \alpha_{t+1}^{(=)}\gamma^{(=)} + \alpha_{t+1}^{(\neq)}\gamma^{(\neq)}$ and $ \beta_{t+1}^{(=)}\gamma^{(=)} + \beta_{t+1}^{(\neq)}\gamma^{(\neq)}$ respectively. Therefore, we have 
    \begin{align}
        \mu_{T}^{\left(\ket{\mathbf{0}} \bra{\mathbf{0}}\right)}&(\ket{p} \bra{q}, \ket{r}\bra{s}) = M_TM_{T-1}\dots M_1 
    \end{align}
    where
    \begin{align}
        M_1 = \begin{bmatrix}
                \alpha_{1} \\
                \beta_{1}
        \end{bmatrix}, \ 
        M_T = \begin{bmatrix}
                \alpha_t & \beta_t
        \end{bmatrix}, \ 
        M_t = \begin{bmatrix}
                \alpha_{t} & \alpha_{t}' \\
                \beta_{t} & \beta_{t}'
        \end{bmatrix}.
    \end{align}
    \begin{align}
        \alpha_t &= \Delta_{T-1}^{(=)} = \frac{\delta_{\text{tr}}^{(T)} + \delta_{\text{ip}}^{(T)}}{2^k + 1}, \ \ \  \beta_t = \Delta_{T-1}^{(\neq)} = \frac{\delta_{\text{tr}}^{(T)}}{2^k + 1},
    \end{align}
    and $2 \leq t \leq T-1$. $\alpha_t$ and $\beta_t$ can be evaluated directly using Lemma~\ref{le:2_design}. Since each $M_t$ is a $2 \times 2$ matrix, its eigenvalues can be computed analytically. We use SymPy for this computation and the eigenvalues are
    \begin{align}
        \frac{\delta_{\text{tr}}^{(t)}}{4}, \ \ \  \frac{\delta_{\text{tr}}^{(t)}\left(4^k-4 \right)}{8 \left(4^k-1 \right)}.
    \end{align}
    We can see that the absolute values of all these eigenvalues are upper bounded by $1/4$.

    Also, we have
    \begin{align*} \label{eq:b_0_gl}
        &\| M_T\|_2 \| M_1\|_2 \\   
        \leq &\sqrt{\frac{10 \cdot 2^{6k} + 140 \cdot 2^{4k} - 240 \cdot 2^{3k} -220 \cdot 2^{2k} + 240 \cdot 2^k + 160}{2^{2k+6}\left( 2^k+1\right)^2\left( 2^{2k}-1\right)^2}} \\
        \leq &1 \numberthis
    \end{align*}
    (computed using SymPy). Combining Eq~\eqref{eq:b_0_gl} with Eq~\eqref{eq:sigma_sigma} gives us
    \begin{align*} \label{eq:sigma_sigma2}
        \mu_{T}^{\left(\ket{\mathbf{0}} \bra{\mathbf{0}} \right)}(\sigma, \sigma) &\leq \sum \limits_{pqrs} | \sigma_{pq} || \sigma_{rs}| \frac{1}{4^{{n-k-1}}} \\
        &= \frac{\| \sigma\|_1^2}{4^{n-k-1}} \numberthis.
    \end{align*}

    Combining Eq~\eqref{eq:sigma_sigma2} with Lemma~\ref{le:equivalence} completes the proof.
\end{proof}

\begin{customthm}{2} 
    Let $\sigma \in \mathbb{D}_n$, $O \coloneqq 1/n\sum_{i=1}^n \ket{0} \bra{0}_i$, and $C_T^{(n)}$ be an MPS ansatz, where each parameterized subcircuit $U_i$ forms a unitary $2$-design. Then, we have 
    \begin{align}
        \text{Var}_{\boldsymbol{\theta}} \left(f_{\sigma, O}(\boldsymbol{\theta}) \right)  \geq \frac{1}{n(2 ^ {2k+1} + 4)} - \frac{h_2(\sigma)}{2n}.
    \end{align}
\end{customthm}

\begin{proof}
    First, notice that for any $i \in \{ 1, \dots, n \}$, 
    \begin{align} \label{eq:z_instead_0}
        &f_{\sigma, \ket{0}\bra{0}_i}(\boldsymbol{\theta}) = \frac{1}{2} + \frac{f_{\sigma, Z_i} (\boldsymbol{\theta}) }{2},
        \end{align}
        implying that
        \begin{align}
        &\text{Var}_{\boldsymbol{\theta}} \left( f_{\sigma, \ket{0}\bra{0}_i}(\boldsymbol{\theta}) \right) =\frac{1}{4} \text{Var}_{\boldsymbol{\theta}} \left( f_{\sigma,Z_i}(\boldsymbol{\theta}) \right). 
    \end{align}

    So, we have
    \begin{align*} \label{eq:variance_to_second_moment}
        \text{Var}_{\boldsymbol{\theta}} \left( f_{\sigma, O}(\boldsymbol{\theta})\right) = &\frac{1}{n^2}\text{Var}_{\boldsymbol{\theta}} \left( \sum \limits_{i=1}^n f_{\sigma, \ket{0}\bra{0}_i}(\boldsymbol{\theta}) \right) \numberthis \\
        = &\frac{1}{4n^2}\text{Var}_{\boldsymbol{\theta}} \left( \sum \limits_{i=1}^n f_{\sigma,Z_i}(\boldsymbol{\theta}) \right) \numberthis \\
        = &\frac{1}{4n^2} \mathbb{E}_{\boldsymbol{\theta}} \left( \sum \limits_{i=1}^n f_{\sigma,Z_i}(\boldsymbol{\theta})\right)^2 \\
        &- \frac{1}{4n^2} \left(\mathbb{E}_{\boldsymbol{\theta}}  \sum \limits_{i=1}^n f_{\sigma, Z_i} (\boldsymbol{\theta})\right)^2 \numberthis \\
        = &\frac{1}{4n^2} \sum \limits_{i=1}^n \mathbb{E}_{\boldsymbol{\theta}} \left(  f_{\sigma,Z_i}(\boldsymbol{\theta}) \right)^2 \\
        &+ \frac{2}{4n^2} \sum \limits_{\substack{i,j=1 \\ i>j}}^n \mathbb{E}_{\boldsymbol{\theta}} \left(  f_{\sigma,Z_i}(\boldsymbol{\theta})  f_{\sigma,Z_j}(\boldsymbol{\theta}) \right) \\
        &- \frac{1}{4n^2} \left(\mathbb{E}_{\boldsymbol{\theta}}  \sum \limits_{i=1}^n f_{\sigma, Z_i}(\boldsymbol{\theta})\right)^2 \numberthis \\
        = &\frac{1}{4n^2} \sum \limits_{i=1}^n \mathbb{E}_{\boldsymbol{\theta}} \left(  f_{\sigma, Z_i}(\boldsymbol{\theta}) \right)^2, \numberthis
    \end{align*}
    where the last equality follow from Lemmas~\ref{le:1_design_multi} and~\ref{le:2_design_multi}.
    
    Similar to the beginning of the proof of Theorem~\ref{th:upper_bound}, for any pure product state $\rho$, using Lemma~\ref{le:equivalence}, we can see that for any $i$
    \begin{align} \label{eq:rho_to_0}
        \mathbb{E}_{\boldsymbol{\theta}} \left(  f_{\rho, Z_i}(\boldsymbol{\theta}) \right)^2 &= \mu_{T}^{(Z_i)}(\rho, \rho) \\
        &= \mu_{T}^{(Z_i)}(\ket{\mathbf{0}} \bra{\mathbf{0}}, \ket{\mathbf{0}} \bra{\mathbf{0}}).
    \end{align}

    Now, we shall derive a lower bound for Eq~\eqref{eq:rho_to_0} $\forall \ i$. We only derive this for $i \leq k$, since we will see that the same lower bound works for any $i > k$ as well. Hence, assume $i \leq k$. First, we compute 
    \begin{align} \label{eq:lower_bound_pq}
        \mu_{T}^{(Z_i)}(\ket{p}\bra{p}, \ket{q}\bra{q}).
    \end{align}
    This will be used later on to compute $ \mu_{T}^{(Z_i)}(\ket{\mathbf{0}} \bra{\mathbf{0}}, \ket{\mathbf{0}} \bra{\mathbf{0}})$.

Our goal is to integrate over $U_T$ and get expressions that look like Eq~\eqref{eq:lower_bound_pq}, but with different bit strings, defined over an $n-1$ qubit system, and as an integral of $T-1$ Haar random gates.

Following the proof of Theorem~\ref{th:upper_bound}, from Eq~\eqref{eq:upper_bound_pqrs}, we have
    \begin{align*} \label{eq:lower_bound_integral_preparation}
        &\mu_{T}^{(Z_i)}\left(\ket{p} \bra{p}, \ket{q}\bra{q} \right) \\
        = &\frac{1}{2^k} \sum \limits_{\substack{i_{1:k} \\ j_{1:k}}} \mu_{T-1}^{(Z_i)}(P_{i_{1:k}} \otimes 
        \ket{p_{k+1:n}} \bra{p_{k+1:n}}, \\
        &\hspace*{2cm}P_{j_{1:k}} \otimes \ket{q_{k+1:n}} \bra{q_{k+1:n}} ) \\
        &\times \frac{1}{2^k} \int \limits_{U_T} \text{tr}\left( P_{i_{1:k}} \ket{p_{1:k}} \bra{p_{1:k}}_{U_T} \right) \\
        &\hspace*{1.4cm}\text{tr} \left( P_{j_{1:k}}\ket{q_{1:k}} \bra{q_{1:k}}_{U_T} \right) \text{d}U_T. \numberthis
    \end{align*}
    We see that whenever $i_1 \neq 0$, the integral drops to $0$. Similar to the proof of Theorem~\ref{th:upper_bound}, we see that when $ i_{2:k} \neq j_{2:k}$, the integral drops to $0$.
    Hence, when $P_{01_2\dots i_k} = P_{0j_2 \dots j_k} \neq \mathds1_{2^k}$, we can directly use Eq~\eqref{eq:upper_bound_pqrs} to get
    \begin{align*} \label{eq:lower_bound_integral_ineq}
        &\frac{1}{2^k} \int \limits_{U_T} \text{tr} \left( P_{0i_{2:k}} \ket{p_{1:k}} \bra{p_{1:k}}_{U_T} \right) \\
        &\hspace*{0.9cm}\text{tr} \left( P_{0i_{2:k}} \ket{q_{1:k}} \bra{q_{1:k}}_{U_T} \right) \text{d}U_T \\
        = &\frac{\delta_{p_{1:k}p_{1:k}}}{2^{2k}-1} - \frac{1}{2^k(2^{2k}-1)} \numberthis\\
        = &\frac{2^k \delta_{q_{1:k}q_{1:k}} - 1}{2^k(2^{2k}-1)}. \numberthis
    \end{align*}

    Similarly, when $P_{01_2\dots i_k} = P_{0j_2 \dots j_k} = \mathds1_{k}$, we have
    \begin{align*} \label{eq:lower_bound_integral_eq}
        \frac{1}{2^k} \int \limits_{U_T} &\text{tr} \left( P_{0i_{2:k}}\ket{p_{1:k}} \bra{p_{1:k}}_{U_T} \right) \\
        &\text{tr} \left( P_{0j_{2:k}}\ket{q_{1:k}} \bra{q_{1:k}}_{U_T} \right) \text{d}U_T = \frac{1}{2^k}. \numberthis \\
    \end{align*}
    Given Eqs~\eqref{eq:lower_bound_integral_eq}, \eqref{eq:lower_bound_integral_ineq}, and~\eqref{eq:lower_bound_integral_preparation},  we have
    \begin{align*}
        &\mu_{T}^{(Z_i)}\left(\ket{p} \bra{p}, \ket{q}\bra{q} \right) \\
        = &\frac{1}{2k-2} \mu_{T-1}^{(Z_i)}(\mathds1_{k-1} \otimes \ket{p_{k+1:n}}\bra{p_{k+1:n}}, \\
        &\hspace*{2cm}\mathds1_{k-1} \otimes \ket{q_{k+1:n}}\bra{q_{k+1:n}}) \\
        &+ \frac{(2^{k+1} \delta_{p_{1:k}q_{1:k}}-2)(2^{2k-2}-1)}{2^{2k-1}(2^{2k}-1)} \\
        &\mu_{T-1}^{(Z_i)}(Z^{\otimes (k-1)} \otimes \ket{p_{k+1:n}}\bra{p_{k+1:n}}, \\
        &\hspace*{0.9cm}Z ^ {\otimes k-1} \otimes \ket{q_{k+1:n}}\bra{q_{k+1:n}}). \numberthis
    \end{align*}
    Let $\mathds1_{k-1} = \sum \limits_{i=0}^{2^{k-1}-1} \ket{i} \bra{i}$ and $ Z^{\otimes (k-1)} = \sum \limits_{i}^{2^{k-1}-1} \lambda_{i} \ket{i} \bra{i}$ be spectral decompositions. Then, we have
    \begin{align*}
    &\mu_{T}^{(Z_i)}\left(\ket{p} \bra{p}, \ket{q}\bra{q} \right) \\
        =&\frac{1}{2^{2k-2}} \sum \limits_{i,j=0} ^ {2^{k-1}-1}\mu_{T-1}^{(Z_i)}(\ket{i}\bra{i} \otimes \ket{p_{k+1:n}}\bra{p_{k+1:n}}, \\
        &\hspace*{2.9cm}\ket{j}\bra{j} \otimes \ket{q_{k+1:n}}\bra{q_{k+1:n}}) \\
        &+ \frac{(2^{k+1} \delta_{p_{1:k}q_{1:k}}-2)(2^{2k-2}-1)}{2^{2k-1}(2^{2k}-1)} \\
        &\hspace*{0.4cm}\sum \limits_{i,j=0} ^ {2^{k-1}-1} \lambda_i \lambda_j \mu_{T-1}^{(Z_i)}(\ket{i}\bra{i} \otimes \ket{p_{k+1:n}}\bra{p_{k+1:n}}, \\
        &\hspace*{3.cm}\ket{j}\bra{j} \otimes \ket{q_{k+1:n}}\bra{q_{k+1:n}}). \numberthis
    \end{align*}
    Next, we define 
    \begin{align*}
        \Delta_t^{(=)} &= \int \limits_{\mathbf{U}_{T-t}} \mu_{T-t}^{(Z_i)}(\ket{0}\bra{0} ^ {\otimes k-1} \otimes \ket{p_{k+t:n}} \bra{p_{k+t:n}}, \\
        &\hspace*{2cm}\ket{0}\bra{0} ^ {\otimes k-1} \otimes \ket{q_{k+t:n}} \bra{q_{k+t:n}}), \numberthis \\
        \Delta_t^{(\neq)} &= \int \limits_{\mathbf{U}_{T-t}} \mu_{T-t}^{(Z_i)}(\ket{0}\bra{0} ^ {\otimes k-1} \otimes \ket{p_{k+t:n}} \bra{p_{k+t:n}}, \\
        &\hspace*{2cm}\ket{1}\bra{1} ^ {\otimes k-1} \otimes \ket{q_{k+t:n}} \bra{q_{k+t:n}}). \numberthis
    \end{align*}
    In a similar manner to how we proceeded in Theorem~\ref{th:upper_bound}, using Lemmas~\ref{le:equivalence} and~\ref{le:cross_product}, we have
    \begin{align}
        \mu_{T}^{(Z_i)}(\ket{p} \bra{p}, \ket{q}\bra{q}) = \alpha \Delta_1^{(=)} + \beta \Delta_1^{(\neq)},
    \end{align}
    where 
    \begin{align}
        \alpha &= \frac{1}{2^{k-1}} + \frac{(2^{2k-2}-1)(2^k\delta_{p_{1:k}q_{1:k}}-1)}{2^{k-1}(2^{2k}-1)}, \\
        \beta &= \frac{2^{k-1}-1}{2^{k-1}} - \frac{(2^{2k-2}-1)(2^k\delta_{p_{1:k}q_{1:k}}-1)}{2^{k-1}(2^{2k}-1)}.
    \end{align}
    Now, let us consider these values when $p=q$ and $p \neq q$. Define
    \begin{align}
        \alpha^{(=)} &= \frac{1}{2^{k-1}} + \frac{2^{2k-2}-1}{2^{k-1}(2^{k}+1)}, \\
        \beta^{(=)} &= \frac{2^{k-1}-1}{2^{k-1}} - \frac{2^{2k-2}-1}{2^{k-1}(2^{k}+1)}, \\
        \alpha^{(\neq)} &= \frac{1}{2^{k-1}} - \frac{2^{2k-2}-1}{2^{k-1}(2^{2k}-1)}, \\ 
        \beta^{(\neq)} &= \frac{2^{k-1}-1}{2^{k-1}} + \frac{2^{2k-2}-1}{2^{k-1}(2^{2k}-1)}.
    \end{align}
    So we have
    \begin{align}
        \mu_{T}^{(Z_i)}(\ket{\mathbf{0}} \bra{\mathbf{0}}, \ket{\mathbf{0}} \bra{\mathbf{0}}) = \alpha^{(=)}\Delta_1^{(=)} + \beta^{(=)}\Delta_1^{(\neq)}.
    \end{align}
    Similar to the proof of Theorem~\ref{th:upper_bound}, we can see that when we integrate $\Delta_1^{(=)}$ and $\Delta_1^{(\neq)}$ with respect to $U_{T-1}$, we get linear combinations of $ \Delta_2^{(=)}$ and $\Delta_2^{(\neq)}$. 

    Assume that after integration over unitaries $U_T, \dots U_{T-t+1}$, we get $ \gamma^{(=)} \Delta_t^{(=)} + \gamma^{(\neq)} \Delta_t^{(\neq)}$. Now, if we integrate over $U_{T-t}$, the coefficients of $\Delta_{t+1}^{(=)}$ and $\Delta_{t+1}^{(\neq)}$ will be $ \alpha^{(=)}\gamma^{(=)} + \alpha^{(\neq)}\gamma^{(\neq)}$ and $ \beta^{(=)}\gamma^{(=)} + \beta^{(\neq)}\gamma^{(\neq)}$ respectively. Unlike in the proof of Theorem~\ref{th:upper_bound}, these coefficients are independent of $t$. 
    
    So we have
    \begin{align}
        \mu_{T}^{(Z_i)}(\ket{\mathbf{0}} \bra{\mathbf{0}}, \ket{\mathbf{0}} \bra{\mathbf{0}}) = M_T M_{n-k-1} M_1
    \end{align}
    where
    \begin{align*}
        M_1 = \begin{bmatrix}
                \alpha \\
                \beta
        \end{bmatrix}, \ 
        M_T = \begin{bmatrix}
                \alpha_T & \beta_T
        \end{bmatrix}, \\ 
        M = \begin{bmatrix}
                \alpha^{(=)} & \alpha^{(\neq)} \\
                \beta^{(=)} & \beta^{(\neq)}
        \end{bmatrix}, \numberthis
    \end{align*}
    \begin{align}
        \alpha_T = \Delta_{T-1}^{(=)}= \frac{1}{2^k + 1},  \ \ \ \ \ \beta_T = \Delta_{T-1}^{(\neq)}
         = \frac{-1}{2^{2k}-1}.
    \end{align}
    $\alpha_t$ and $\beta_t$ can be evaluated directly using Lemma~\ref{le:2_design}. The eigenvalues of the matrix $M$ are $1$ and
    \begin{align}
        B_2 = \frac{2^{2k-2} - 1}{2^{2k}-1}.
    \end{align}

    Therefore,
    \begin{align*} \label{eq:lower_bound}
        \mu_{T}^{(Z_i)}(\ket{\mathbf{0}} \bra{\mathbf{0}}, \ket{\mathbf{0}} \bra{\mathbf{0}}) &= B_0 + B_1 B_2^{n-k-1} \\
        &\geq B_0, \numberthis
    \end{align*}
    where
    \begin{align}
        B_0 &= \frac{1}{2 ^ {2k-1} + 1}, \\
        B_1 &= \frac{2^{5k} - 2^{4k} -6 \cdot 2 ^ {3k} + 4 \cdot 2^{2k} + 2 \cdot 2^k }{2(2^k + 1)^2\cdot (2^{2k}-1)(2^{2k} + 2)}.
    \end{align}

    Now, notice that whenever $i > k$, we have 
    \begin{align}
        &\int \limits_{\mathbf{U}_T}\left( \text{tr}\left(Z_i{\ket{\mathbf{0}} \bra{\mathbf{0}}}_{C_T} \right) \right)^2 \text{d}\mathbf{U}_T \\
        = &\int \limits_{\mathbf{U}_{T-i+k}} \text{tr}\left( Z_k \ket{0}\bra{0}^{\otimes n-i+k}_{C_{T-i+k}} \right)  ^ 2 \text{d}\mathbf{U}_{T-i+k}, \numberthis
    \end{align}
    which is also an instance of the previous case defined over $n-i+k$ qubits. Hence, from Eq~\eqref{eq:lower_bound}, we can see that the same lower bound shall apply in this case as well. So, we have
    \begin{align}
        \text{Var}_{\boldsymbol{\theta}} \left( f_{\rho, O}(\boldsymbol{\theta})\right) \geq \frac{1}{4n} B_0.
    \end{align}

    Now, consider $\sigma \in \mathbb{D}_n$ such that $\|\rho - \sigma\|_{\text{tr}} \leq \epsilon$ for some pure product state $\rho$. Then, for any $i$, we have
    \begin{align*}
        &\left|\int \limits_{\mathbf{U}_T} \text{tr}\left(Z_i\rho_{ C_T} \right) ^ 2 \text{d}\mathbf{U}_T - \int \limits_{\mathbf{U}_T} \text{tr}\left( Z_i \sigma_{C_T} \right) ^ 2 \text{d}\mathbf{U}_T \right| \\
        \leq &\int \limits_{\mathbf{U}_T}\left| \text{tr}\left( Z_i\rho_{C_T} \right) ^ 2 - \text{tr}\left(Z_i \sigma_{C_T} \right) ^ 2 \right| \text{d}\mathbf{U}_T \numberthis \\
        = &\int \limits_{\mathbf{U}_T}\left| \text{tr} \left( Z_i\rho_{C_T} \right) - \text{tr} \left( Z_i\sigma_{C_T} \right) \right| \\
        &\left| \text{tr}\left( Z_i\rho_{C_T} \right) + \text{tr} \left( Z_i \sigma_{C_T} \right)\right| \text{d}\mathbf{U}_T \numberthis \\
        = &\int \limits_{\mathbf{U}_T} \left| \text{tr} \left( Z_i(\rho-\sigma)_{C_T} \right) \right| \left|\text{tr} \left( Z_i(\rho + \sigma)_{C_T} \right) \right|\text{d}\mathbf{U}_T \numberthis \\
        \leq &\| \rho - \sigma\|_{\text{tr}} \| \| \rho + \sigma\|_{\text{tr}} \numberthis  \\
        \leq & 2 \| \rho - \sigma\|_{\text{tr}} \numberthis.
    \end{align*}
    where $ \| \cdot \|_{\text{tr}}$ is the trace norm and the second-last inequality follows from Tracial Matrix H\"{o}lder's Inequality (Lemma~\ref{le:matrix_holder}). Then, we have
    \begin{align*}
        &\sum \limits_{i=1}^n \int \limits_{\mathbf{U}_T} \text{tr}\left( Z_i\rho_{C_T} \right) ^ 2 \text{d}\mathbf{U}_T -2n\| \sigma - \rho\|_{\text{tr}} \\
        \leq &\sum \limits_{i=1}^n \int \limits_{\mathbf{U}_T} \text{tr}\left( Z_i\sigma_{C_T} \right) ^ 2 \text{d}\mathbf{U}_T, \numberthis 
        \end{align*}
        implying that
        \begin{align*}
        &\frac{1}{4n^2}\sum \limits_{i=1}^n \int \limits_{\mathbf{U}_T} \text{tr}\left( Z_i\rho_{C_T}  \right) ^ 2 \text{d}\mathbf{U}_T -\frac{1}{2n}\| \sigma - \rho\|_{\text{tr}} \\
        \leq &\frac{1}{4n^2} \sum \limits_{i=1}^n \int \limits_{\mathbf{U}_T} \text{tr}\left( Z_i\sigma_{C_T} \right) ^ 2 \text{d}\mathbf{U}_T \numberthis
        \end{align*}
        further implying that
        \begin{align*}
        &\text{Var}_{\boldsymbol{\theta}} \left( f_{\rho,O}(\boldsymbol{\theta})\right) -\frac{1}{2n}\| \sigma - \rho\|_{\text{tr}} \leq  \text{Var}_{\boldsymbol{\theta}} \left( f_{\sigma, O}(\boldsymbol{\theta})\right), \numberthis
        \end{align*}
        and hence, we have
        \begin{align*}
        &\frac{B_0}{4n} -\frac{\epsilon}{2n} \leq \text{Var}_{\boldsymbol{\theta}} \left( f_{\sigma,O}(\boldsymbol{\theta})\right) \numberthis.
    \end{align*}

    Now, for this $\sigma$, minimization over all $\rho$ as required by the theorem completes the proof.
    \end{proof}

\begin{customco}{1}
    Let $\sigma \in \mathbb{D}_n$ and $C_T^{(n)}$ be an MPS ansatz. Then, $f_{\sigma, \ket{\mathbf{0}} \bra{\mathbf{0}}}$ exhibitis barren plateaus if $ h_1(\sigma) \in \mathcal{O}(4^{n/p})$.
\end{customco}
\begin{proof}

    Throughout this proof, we assume that within the computation of the variance, $U_p^{(L,q)}$ is distributed according to the Haar measure since trivial changes to the proof are sufficient to prove the same when $U_p^{(R,q)}$ is distributed according to Haar measure. Let $\mathbf{U}_T^{(p,q)} = \{U_1, \dots, U_{p-1}, \theta_{pq}, U_p^{(L,q)}, U_{p+1}, \dots, U_T\}$. Using Lemmas~\ref{le:exp_partial_derivative} and~\ref{le:parameter_shift_rule}, we have
    \begin{align*} \label{eq:upper_bound_corrollary}
         &\text{Var}_{\boldsymbol{\theta}} \left(\partial_{\theta_{pq}}f_{\sigma,\ket{\mathbf{0}}\bra{\mathbf{0}}}(\boldsymbol{\theta})\right) \\
         = &\mathbb{E}_{\boldsymbol{\theta}} \left(\partial_{\theta_{pq}}f_{\sigma,\ket{\mathbf{0}} \bra{\mathbf{0}}}(\boldsymbol{\theta})\right)^2 \numberthis \\
         = &\frac{1}{4} \mathbb{E}_{\boldsymbol{\theta}} \left( f_{\sigma, \ket{\mathbf{0}} \bra{\mathbf{0}}} (\boldsymbol{\theta}_{pq+}) - f_{\sigma, \ket{\mathbf{0}} \bra{\mathbf{0}}} (\boldsymbol{\theta}_{pq-}) \right) ^ 2 \numberthis \\
         = &\frac{1}{4} \mathbb{E}_{\boldsymbol{\theta}} \Big( \left( f_{\sigma, \ket{\mathbf{0}} \bra{\mathbf{0}}} (\boldsymbol{\theta}_{pq+}) \right)^2 + \left( f_{\sigma, \ket{\mathbf{0}} \bra{\mathbf{0}}} (\boldsymbol{\theta}_{pq-}) \right)^2 \\
         &- 2f_{\sigma, \ket{\mathbf{0}} \bra{\mathbf{0}}} (\boldsymbol{\theta}_{pq+}) f_{\sigma, \ket{\mathbf{0}} \bra{\mathbf{0}}} (\boldsymbol{\theta}_{pq-}) \Big) \numberthis \\
         \leq &\frac{1}{4} \mathbb{E}_{\boldsymbol{\theta}} \left( \left(f_{\sigma, \ket{\mathbf{0}} \bra{\mathbf{0}}} (\boldsymbol{\theta}_{pq+}) \right)^2 + \left(f_{\sigma, \ket{\mathbf{0}} \bra{\mathbf{0}}} (\boldsymbol{\theta}_{pq-}) \right)^2 \right) \\
         &+ \frac{1}{2}\left| \mathbb{E}_{\boldsymbol{\theta}} \left(f_{\sigma, \ket{\mathbf{0}} \bra{\mathbf{0}}} (\boldsymbol{\theta}_{pq+}) f_{\sigma, \ket{\mathbf{0}} \bra{\mathbf{0}}} (\boldsymbol{\theta}_{pq-}) \right) \right|. \numberthis
    \end{align*}
    Now, notice that 
    \begin{align*} \label{eq:squared_pi}
        \mathbb{E}_{\boldsymbol{\theta}} \left(f_{\sigma, \ket{\mathbf{0}} \bra{\mathbf{0}}} (\boldsymbol{\theta}_{pq+})\right)^2 &= \mathbb{E}_{\boldsymbol{\theta}} \left(f_{\sigma, \ket{\mathbf{0}} \bra{\mathbf{0}}} (\boldsymbol{\theta}_{pq-})\right)^2 \\
        &= \mathbb{E}_{\boldsymbol{\theta}} \left(f_{\sigma, \ket{\mathbf{0}} \bra{\mathbf{0}}} (\boldsymbol{\theta})\right)^2. \numberthis
    \end{align*}
    This follows from combining Lemma~\ref{le:eta} with the fact that $U_p(\boldsymbol{\theta}_{pq \pm }) = U_p^{(L,q)}({\boldsymbol{\theta}_p}) e^{-i (\theta_{pq} \pm \pi/2) H_{pq}} U_p^{(R,q)}({\boldsymbol{\theta}_p}) = U_p^{(L,q)}({\boldsymbol{\theta}_p}) e^{-i \theta_{pq} H_{pq}} e^{\pm \frac{i \pi H_{pq}}{2} } U_p^{(R,q)}({\boldsymbol{\theta}_p})$.
    Also,
    \begin{align*} \label{eq:cross_product_pi}
        &\Bigg| \int \limits_{\mathbf{U}_T^{(\theta)}} f_{\sigma, \ket{\mathbf{0}} \bra{\mathbf{0}}} (\boldsymbol{\theta}_{pq+})f_{\sigma, \ket{\mathbf{0}} \bra{\mathbf{0}}} (\boldsymbol{\theta}_{pq-}) \\
        &\left( \prod \limits_{j = 1, j \neq p}^{T} \text{d}U_j \right)\text{d}U_p^{(L,q)} \text{d}\theta_{pq} \Bigg| \\
        \leq &\sqrt{\int \limits_{\mathbf{U}_T^{(\theta)}} \left(f_{\sigma, \ket{\mathbf{0}} \bra{\mathbf{0}}} (\boldsymbol{\theta}_{pq+}) \right) ^ 2 \left( \prod \limits_{j = 1, j \neq p}^{T} \text{d}U_j \right)\text{d}U_p^{(L,q)} \text{d}\theta_{pq}} \\
        &\sqrt{\int \limits_{\mathbf{U}_T^{(\theta)}} \left( f_{\sigma, \ket{\mathbf{0}} \bra{\mathbf{0}}} (\boldsymbol{\theta}_{pq-}) \right)^2 \left( \prod \limits_{j = 1, j \neq p}^{T} \text{d}U_j \right)\text{d}U_p^{(L,q)} \text{d}\theta_{pq}} \numberthis \\
        = &\mathbb{E}_{\boldsymbol{\theta}} \left(f_{\sigma, \ket{\mathbf{0}} \bra{\mathbf{0}}} (\boldsymbol{\theta})\right)^2. \numberthis
    \end{align*}
    This follows directly from Cauchy Schwarz inequality. Plugging Eqs~\eqref{eq:cross_product_pi} and~\eqref{eq:squared_pi} into Eq~\eqref{eq:upper_bound_corrollary} completes the proof.
\end{proof}

    \begin{customco}{2} 
        Let $\sigma \in \mathbb{D}_n$, $C_T^{(n)}$ be an MPS ansatz, and $O \coloneqq 1/n\sum \limits_{i=1}^n \ket{0}\bra{0}_i$. Then, $f_{\sigma, O}$ exhibitis barren plateaus if $ h_2(\sigma) \ll  1/(2 ^ {2k} + 2)$.    \end{customco}
    \begin{proof}

    First, we recall an important lemma relating cost concentration with barren plateaus from~\cite{Arrasmith2022}.
        \begin{lemma} \cite{Arrasmith2022} \label{th:cost_concentration}
        For any ansatz $C(\boldsymbol{\theta})$, $\sigma \in \mathbb{D}_n$ and $ W \in \mathbb{H}_n$, if $ \ \forall \  p,q$, where $ 1 \leq p \leq T, 1 \leq q \leq m$,
        \begin{align}
            \text{Var}_{\boldsymbol{\theta}} \left( \partial_{pq} f_{\sigma,W}(\boldsymbol{\theta}) \right) \in \mathcal{O}(1/b^n)
        \end{align}
        for some $b > 1$, then 
        \begin{align}
            \text{Var}_{\boldsymbol{\theta}} \left( f_{\sigma,W}(\boldsymbol{\theta})\right) \in \mathcal{O}(1/b^n).
        \end{align}.
    \end{lemma}

    The result then follows from the contrapositive of Lemma~\ref{th:cost_concentration}.
    \end{proof}

\begin{customthm}{3}
    For any $n$-qubit quantum circuit $V$, let $R_V = \max_i R_{V,i}$, where $R_{V,i}$ is the number of 2-qubit gates being applied on any qubits $j, k$ such that $j \leq i \leq k$. Then, for any product observable $W$, $ \| W_{V^{\dag}}\|_{\mathbb{K}}$ can be classically evaluated with cost $\mathcal{O}(2^{R_V})$.
\end{customthm}
\begin{proof}
    From~\cite{Jozsa2006}, we can see that given a classical description of a circuit $V$, an MPS description of $W_{V^{\dag}}$ with bond dimension $\mathcal{O}(2^{R_V})$ can be computed using computational cost scaling as $\mathcal{O}(2^{R_V}) $. The change from the standard basis to the orthonormal Pauli basis is efficient, involving only local rotations. Let $\widehat{W}_{V^{\dag}} $ be a vector of Pauli basis coefficients of $ W_{V^{\dag}}$. Then, $ \| W\|_{\mathbb{K}} = \sqrt{\| \widehat{W}_{V^{\dag}} * \widehat{W}_{V^{\dag}}\|_2/2^n}$, where $*$ is the Hadamard product. From~\cite{Oseledets2011}, the Hadamard and inner products of tensors represented as MPS can be computed efficiently, with cost scaling polynomially in their bond dimension, thus completing the proof.
\end{proof}

\begin{customthm}{4}
    For any parameterized circuit $C$, and an orthonormal basis $\mathbb{K}$ of $\mathbb{C}^{2^n \times 2^n}$, define 
    \begin{align} \label{eq:C_norm}
        \| W\|_{C,\mathbb{K}} \coloneqq \int \limits_{\boldsymbol{\theta}} \left\Vert W_{{C(\boldsymbol{\theta})}^{\dag}}\right\Vert_{\mathbb{K}} \text{d} \boldsymbol{\theta},
    \end{align}
    for any $W \in \mathbb{H}_n$. Then, $\| \cdot \|_{C,\mathbb{K}}$ is a norm on $\mathbb{H}_n$.
\end{customthm}
\begin{proof}
    First notice that $\| W \|_{C,\mathbb{K}}$ is non-negative since it is an average of norms. When $W = 0$, then $\| W \|_{C,\mathbb{K}} = 0$. Similarly, when $\| W \|_{C,\mathbb{K}} = 0$, we will have $\| W_{C(\boldsymbol{\theta})^{\dag}}\|_{\mathbb{K}} = 0 \ \forall \ \boldsymbol{\theta}$ and hence $W = 0$.

    Next, we prove the triangle inequality.
    \begin{align}
        &\| W^{(1)} \|_{C,\mathbb{K}} + \| W^{(2)} \|_{C,\mathbb{K}} \\
        = &\int \limits_{\boldsymbol{\theta}} \left\Vert W^{(1)}_{C(\boldsymbol{\theta})^{\dag}}\right\Vert_{\mathbb{K}} + \left\Vert W^{(2)}_{C(\boldsymbol{\theta})^{\dag}}\right\Vert_{\mathbb{K}} \text{d} \boldsymbol{\theta} \\
        \geq &\int \limits_{\boldsymbol{\theta}} \left\Vert  W^{(1)}_{C(\boldsymbol{\theta})^{\dag}} + W^{(2)}_{C(\boldsymbol{\theta})^{\dag}} \right\Vert_{\mathbb{K}} \text{d} \boldsymbol{\theta} \\
        = &\int \limits_{\boldsymbol{\theta}} \left\Vert \left(W^{(1)} + W^{(2)} \right)_{C(\boldsymbol{\theta})^{\dag}} \right\Vert_{\mathbb{K}} \text{d} \boldsymbol{\theta} \\
        = &\| W^{(1)} + W^{(2)}\|_{C,\mathbb{K}}.
    \end{align}
\end{proof}








\end{document}